\documentclass[oneside, 12pt]{article}
\usepackage{geometry} 
\usepackage{helvet}
\usepackage{amsmath}
\usepackage{amssymb}
\usepackage{epstopdf}
\usepackage{amsthm}
\usepackage{amsfonts}
\usepackage{graphicx}
\usepackage{listings}
\usepackage{color}
\usepackage{psfrag,graphics,epsfig,multirow,multicol,natbib}
\usepackage{setspace}

\title{Non-Gaussian Quasi Maximum Likelihood Estimation of GARCH Models}
\author{Lei Qi\footnote { \small{Lei Qi and Dacheng Xiu are PhD candidates, Jianqing Fan is Frederick L. Moore'18 Professor of Finance at
Bendheim Center for Finance, Princeton University, Princeton, NJ 08544. (Email: lqi@princeton.edu, dachengx@princeton.edu, jqfan@princeton.edu) This research was supported in part by NSF Grant DMS-0704337 and DMS-0714554.}}, Dacheng Xiu and Jianqing Fan}

\newtheorem{Theorem}{Theorem}
\newtheorem{Assumption}{Assumption}
\newtheorem{Remark}{Remark}
\newtheorem{Lemma}{Lemma}
\newtheorem{Proposition}{Proposition}

\newcommand{\BSigma}{\boldsymbol{\Sigma}}
\newcommand{\Btheta}{\boldsymbol{\theta}}
\newcommand{\Bgamma}{\boldsymbol{\gamma}}
\newcommand{\BW}{\boldsymbol{W}}
\newcommand{\BT}{\boldsymbol{T}}
\newcommand{\ve}{\varepsilon}

\newcommand{\By}{\boldsymbol{y}}
\newcommand{\BV}{\boldsymbol{V}}
\newcommand{\BG}{\boldsymbol{G}}
\newcommand{\BXi}{\boldsymbol{\Xi}}

\begin{document}

\maketitle

\begin{abstract}
The non-Gaussian quasi maximum likelihood estimator is frequently used in GARCH models with intension to improve the efficiency of the GARCH parameters.  However, the method is usually inconsistent unless the quasi-likelihood happens to be the true one.  We identify an unknown scale parameter that is critical to the consistent estimation of non-Gaussian QMLE.  As a part of estimating this unknown parameter, a two-step non-Gaussian QMLE (2SNG-QMLE) is proposed for estimation the GARCH parameters. Without assumptions on symmetry and unimodality of the distributions of innovations, we show that the non-Gaussian QMLE remains consistent and asymptotically normal, under a general framework of non-Gaussian QMLE.  Moreover, it has higher efficiency than the Gaussian QMLE, particularly when the innovation error has heavy tails.  Two extensions are proposed to further improve the efficiency of 2SNG-QMLE.  The impact of relative heaviness of tails of the innovation and quasi-likelihood distributions on the asymptotic efficiency has been thoroughly investigated.  Monte Carlo simulations and an empirical study confirm the advantages of the proposed approach.
\end{abstract}


\section{Introduction}

Volatility has been a crucial ingredient in modeling financial time series and designing risk management and trading strategies. It is often observed that volatilities tend to cluster together. This characteristic of financial data suggests that volatilities are autocorrelated and changing over time. \cite{Engle1982} proposed ARCH (autoregressive conditional heteroscedasticity) to model volatility dynamics by taking weighted averages of past squared forecast errors. This seminal idea led to a great richness and variety of volatility models. Among numerous generalizations and developments, GARCH model by \cite{Bollerslev1986} has been commonly used:
\begin{eqnarray}
&&x_t= v_t \varepsilon_t\\
&&v^2_t= c +\sum_{i=1}^p \tilde a_i x_{t-i}^2+\sum_{j=1}^q \tilde b_j v_{t-j}^2
\end{eqnarray}
In this GARCH$(p,q)$ model, variance forecast takes weighted average of not only past square errors but also historical variances. The simplicity and intuitive appeal make GARCH model, especially GARCH$(1,1)$, a workhorse and good start point in many financial applications.

Earlier literature on inference from ARCH/GARCH models is based on the maximum likelihood estimation (MLE) with conditional Gaussian assumption on the innovations. Plenty of empirical evidence, however, has documented heavy-tailed and asymmetric innovation distributions of $\varepsilon_t$, rendering this assumption unjustified, see for instance \cite{Diebold1988}. Consequently, MLE using Student's $t$ or generalized Gaussian likelihood functions has been introduced, see e.g. \cite{EngleBollerslev1986}, \cite{Bollerslev1987},  \cite{Hsieh1989}, and \cite{Nelson1991}. However, these methods may lead to inconsistent estimates if the distribution of the innovation is misspecified. Alternatively, the Gaussian MLE, regarded as a quasi maximum likelihood estimator (QMLE) may be consistent, see e.g. \cite{Elie_Jeantheau_1995}, and asymptotically normal, provided that the innovation has a finite fourth moment, even if it is far from Gaussian, see \cite{HY2003} and \cite{BHK_2003}. The asymptotic theory dates back to as early as \cite{Weiss_1986} for ARCH models, \cite{Lee_Hansen_1994} and \cite{Lumsdaine_1996} for GARCH$(1,1)$ with stronger conditions, and \cite{BW1992} for GARCH$(p,q)$ under high level assumptions.

Nevertheless, gain in robustness comes with efficiency loss. Theoretically, the divergence of Gaussian likelihood from the true innovation density may considerably increase the variance of the estimates, which thereby fails to reach the Cram\'{e}r-Rao bound by a wide margin, reflecting the cost of not knowing the true innovation distribution. \cite{Engle_Gonzalez-Rivera1991} has suggested a semiparametric procedure that can improve the efficiency of the parameter estimates up to 50\% over the QMLE based on their Monte Carlo simulations, but still incapable of capturing the total potential gain in efficiency, see also \cite{Linton1993}. \cite{DrostKlaassen1997} has put forward an adaptive two-step semiparametric procedure based on a re-parametrization of the GARCH$(1,1)$ model with unknown but symmetric error. \cite{GonzlezRivera_Drost1999} has compared its efficiency gain/loss over Gaussian QMLE and MLE. All the effort would become void if the innovation fails to have a finite fourth moment. \cite{HY2003} has considered the Gaussian QMLE and shown that it would converge to stable distributions asymptotically rather than a normal distribution.

The empirical reason of Gaussian QMLE's efficiency loss is that financial data are generally heavy tail distributed. The conditional normality assumption is violated. For example, \cite{BW1992} reported that sample kurtosis of estimated residuals of Gaussian QMLE on S\&P500 monthly data is 4.6, well exceeding the Gaussian kurtosis which is 3. It is therefore intuitively appealing to develop QMLE based on non-Gaussian likelihoods, especially heavy tailed likelihoods. And the efficiency loss of Gaussian QMLE can be greatly reduced by replacing the likelihoods with heavy tailed ones.

In contrast with the majority of literature focusing on Gaussian QMLE for inference, there is rather limited attention on inference using non-Gaussian QMLE. This may be partly due to the fact that the Gaussian QMLE is robust against misspecification of error distribution, while directly using non-Gaussian QMLE is not. In general a non-Gaussian QMLE does not yield consistent estimation when true error distribution deviates from the likelihood. Moreover, this inconsistency could not be corrected even as we allow to estimate a shape parameter indexing the non-Gaussian likelihood family together with model parameters unless the true innovation density is a member of this likelihood family. Otherwise, estimating shape along with model parameters simply picks one likelihood that is ``least" biased, however the bias persists. \cite{Newey1997} have considered the identification of the non-Gaussian QMLE for heteroscedastic parameters in general conditional heteroscedastic models. They have also pointed out that the scale parameter may not be identified as its true value since it is no longer a natural scale parameter for non-Gaussian densities.

 A valid remedy served for non-Gaussian QMLE would be manipulating model assumptions in order to maintain consistent estimation. For example, the true innovation density is sometimes taken to be Student's $t$ or generalized Gaussian for granted. Alternatively, \cite{Berkes2004} has shown that with a different moment condition on the true innovations instead of the original $E(\varepsilon^2)=1$, a corresponding non-Gaussian QMLE would obtain consistency and asymptotic normality. However, this moment condition $E(\varepsilon^2)=1$ is an essential assumption which enables $v_t$ to bear the natural interpretation of the conditional standard deviation, the notion of volatility. More importantly, moment condition is part of model specification, and it should be prior to and independent of the choice of likelihood. Changing the moment condition would not solve the robustness issue of non-Gaussian QMLE; it simply renders consistency to the correct combination of moment condition and non-Gaussian likelihood, which cannot be determined without knowing the true innovation.

 Therefore, we prefer a non-Gaussian QMLE method which is robust against error misspecification, more efficient than Gaussian QMLE, independent of model assumptions, and yet practical. Such method can also well extend the usage of non-Gaussian QMLE in GARCH software packages. Current packages do include choice of likelihood as an option, for example, Student's $t$ and generalized Gaussian. In addition the shape parameter can be specified or estimated. But as discussed before, such method is not robust against error misspecification. When running estimation, one chooses a particular likelihood family with the hope that true innovation distribution falls into such family, but typically it does not.

The main contribution of this paper is that we propose a novel two step non-Gaussian QMLE method, 2SNG-QMLE for short, which meets the desired properties. The key is the estimation of a scale adjustment parameter, denoted as $\eta_f$, for non-Gaussian likelihood to ensure the consistency of non-Gaussian QMLE under any error distributions. $\eta_f$ is estimated through Gaussian QMLE in the first step; then we feed the estimated $\eta_f$ into non-Gaussian QMLE in the second step. In Gaussian QMLE $\eta_f$ is held constant at unity, and partly because of that this quantity has been overlooked; but in non-Gaussian QMLE $\eta_f$ is no longer constant, and how much it deviates from unity measures how much asymptotic bias would incur by simply using non-Gaussian QMLE without such adjustment.

The second contribution is that we adopt a re-parameterized GARCH model (see also \cite{Newey1997} and \cite{Drost_Klaassen1997}) which separates the volatility scale parameter from heteroscedastic parameters. Under this new parametrization, we derive asymptotic behaviors for 2SNG-QMLE. The results show that 2SNG-QMLE is more efficient than Gaussian-QMLE under various innovation settings, and furthermore there is a clear cut on asymptotic behaviors under the new parametrization. The heteroscedastic parameters we can always achieve $T^{1\over2}$ asymptotic normality, whereas Gaussian QMLE has slower convergence rate when error does not have fourth moment.

The outline of the paper is as follows. Section 2 introduces the model and assumptions. Section 3 discusses the estimation procedure and derives the asymptotic results for 2SNG-QMLE. Section 4 proposes two extensions to further improve efficiency. Section 5 employs Monte Carlo simulations to verify the theoretic results. Section 6 conducts real data analysis on stock returns. Section 7 concludes. The appendix provides all the mathematical proofs.

\section{The Model and Assumptions}
The re-parameterized GARCH$(p, q)$ model takes on the following parametric form:
\begin{eqnarray}
&&x_t=\sigma v_t \varepsilon_t\\
&&v^2_t=1+\sum_{i=1}^p a_i x_{t-i}^2+\sum_{j=1}^q b_j v_{t-j}^2
\end{eqnarray}
The model parameters are summarized in $\boldsymbol{\theta} = \{\sigma, \boldsymbol{\gamma}'\}'$, where $\sigma$ is the scale parameter and $\boldsymbol{\gamma}=(\boldsymbol{a}', \boldsymbol{b}')'$ is the autoregression parameter. The true parameter $\boldsymbol{\theta_0}$ is in the interior of $\Theta$, which is a compact subset of the $\boldsymbol{R}_+^{1+p+q}$, satisfying $\sigma>0$, $a_i\geq 0$, $b_j\geq 0$. We use subscript 0 to denote the value under the true model throughout the paper. The innovation $\{\varepsilon_t\}_{t=1,\ldots,T}$ are i.i.d random variables with mean 0, variance 1 and unknown density $g(\cdot)$. In addition, we assume that the GARCH process $\{x_t\}$ is strictly stationary and ergodic. The elementary conditions for the stationarity and ergodicity of GARCH models have been discussed in \cite{Bougerol_Picard_1992}.

We consider a parametric family of quasi likelihood $\{\eta: \frac{1}{\eta}f(\frac{\cdot}{\eta})\}$ indexed by $\eta>0$, for any given likelihood function $f$. Unlike a shape parameter that is often included in a Student's $t$ likelihood function, $\eta$ is a scale parameter selected to reflect the penalty of model misspecification. More precisely, a specific quasi likelihood scaled by $\eta_f$ will be used in the estimation procedure. The parameter $\eta_f$ minimizes the discrepancy between the true innovation density $g$ and an unscaled misspecified quasi likelihood in the sense of Kullback Leibler Information Distance, see e.g. White(1982). Or equivalently,
\begin{eqnarray}\label{eq:eta_likelihood}\eta_f=\mbox{argmax}_{\eta>0}E\big\{-\log \eta+\log
f(\frac{\varepsilon}{\eta})\big\}\end{eqnarray}
where the expectation is taken under the true model $g$.

Note that $\eta_f$ here only depends on the divergence of the two densities under consideration, rendering it a universal measure of closeness irrelevant of the GARCH model. Once $\eta_f$ is given, the QMLE $\hat{\boldsymbol{\theta}}$ is defined by maximizing the following quasi likelihood with this model parameter $\eta_f$:
\begin{eqnarray}\label{eq:likelihood}
L_T(\boldsymbol{\theta})=\frac{1}{T}\sum_{t=1}^T
l_t(\boldsymbol{\theta})=\frac{1}{T}\sum_{t=1}^T\big(-\log(\sigma v_t)+\log
f(\frac{x_t}{\eta_f\sigma v_t})\big)
\end{eqnarray}

Apparently, the likelihood function differs from a regular one with the additional model parameter $\eta_f$. In fact, our approach is a generalization of the Gaussian QMLE and the MLE as illustrated in the next proposition.
\begin{Proposition}\label{cor:example} If $f\propto \exp(-{x^2}/{2})$ or $f=g$, then $\eta_f=1$.
\end{Proposition}

Moreover, it can be implied from \cite{Newey1997} that in general, an unscaled non-Gaussian likelihood function applied in this new re-parametrization of GARCH$(p,q)$ setting fails to identify the volatility scale parameter $\sigma$, resulting in inconsistent estimates. We show in the next section that incorporating $\eta_f$ into the likelihood function facilitates the identification of the volatility scale parameter.

For convenience, we assume the following regularity conditions are always satisfied: $f$ is twice continuously differentiable, and for any $\eta>0$, we have $\sup_{\boldsymbol{\theta}\in\Theta}E|l_t(\theta)|<\infty$, $E\sup_{\boldsymbol{\theta}\in\mathcal{N}}|\nabla l_t(\boldsymbol{\theta})|<\infty$, and $E\sup_{\boldsymbol{\theta}\in\mathcal{N}}|\nabla^2 l_t(\boldsymbol{\theta})|<\infty$, for some neighborhood $\mathcal{N}$ of $\boldsymbol{\theta_0}$.

\section{Main Results}
\subsection{Identification}

Identification is a critical condition for consistency. It requires that the expected quasi likelihood
$\bar{L}_T(\boldsymbol{\theta})=E(L_T(\boldsymbol{\theta}))$ has a unique maximum at the true
parameter value $\boldsymbol{\theta_0}$. To show that $\boldsymbol{\theta}$ can be identified in the presence of $\eta_f$, we make the following assumptions:
\begin{Assumption}\label{ASS1}A quasi likelihood of the GARCH $(p,q)$ model is selected such that
\begin{enumerate}
\item $v_t(\boldsymbol{\gamma_0})>0$, and $v_t(\boldsymbol{\gamma})/v_t(\boldsymbol{\gamma_0})$ is
not a constant if $\boldsymbol{\gamma}\neq \boldsymbol{\gamma_0}$.
\item The function $Q(\eta)=-\log \eta+E(\log
f(\frac{\varepsilon}{\eta}))$ has a unique maximizer $\eta_f>0$.
\end{enumerate}
\end{Assumption}

Note that the first point is the usual identification condition for the autoregression parameter $\boldsymbol{\gamma_0}$, and that the second requirement is the key to the identification of $\sigma_0$.

\begin{Lemma}\label{lem:identification}
Given Assumption \ref{ASS1}, $\bar{L}_T(\boldsymbol{\theta})$ has
a unique maximum at the true value $\boldsymbol{\theta}=\boldsymbol{\theta_0}$.
\end{Lemma}

The next lemma provides a few primitive sufficient conditions for the last statement of
Assumption \ref{ASS1}. The conditions given below provide a general guideline of choosing an appropriate quasi likelihood $f$.

\begin{Lemma}{\label{lem:primitive}}
Assume that $f$ is continuously differentiable up to the second order
and $h(x)=x\frac{\dot{f}(x)}{f(x)}$. Suppose that
$\{\varepsilon_t\}\sim\varepsilon$ is i.i.d. with mean $0$,
variance $1$ and a finite $p^{th}$ moment. If, in addition,
\begin{enumerate}
  \item $h(x)\leq 0$.
   \item $x\dot{h}(x)\leq 0$ and the equality holds if and only if $x=0$.
  \item $|h(x)|\leq C|x|^p$, and $|x\dot{h}(x)|\leq C|x|^p$, for some constant $C>0$, and $p\geq 0$.
  \item $\displaystyle\limsup_{x\rightarrow \infty}h(x)<-1$
\end{enumerate}
then $Q(\eta)$ has a unique maximum at some point $\eta_f>0$. Furthermore, $\eta_f>1$ if and only if $Eh({\varepsilon})<-1$.
\end{Lemma}

The last three assumptions are more general than the concavity of the function $Q(\eta)$. For some commonly used likelihood such as the family of Student's $t$ likelihood, the concavity assumption of $Q$ is violated. However, it still satisfies the above lemma. A few examples of families of likelihood that satisfy Lemma \ref{lem:primitive} are given below.

\begin{Remark}\label{rmk:likelihood}
If $f=\frac{1}{\sqrt{2\pi}}e^{-\frac{x^2}{2}}$, then $h(x)=-x^2$. If
$f$ is the standardized $t_\nu$-distribution with $\nu>2$, that is
$f\propto (1+\frac{x^2}{\nu-2})^{-\frac{\nu+1}{2}}$, then
$h(x)=-\frac{(\nu+1) x^2}{\nu-2+x^2}$. Both cases satisfy Lemma
\ref{lem:primitive} with $p=2$. In addition, if  $\log f(x) =
-|x|^\beta(\frac{\Gamma(\frac{3}{\beta})}{\Gamma(\frac{1}{\beta})})^{\frac{\beta}{2}}
+ \mbox{const}$, the generalized Gaussian likelihood, then
$h(x)=-\beta(\frac{\Gamma(\frac{3}{\beta})}{\Gamma(\frac{1}{\beta})})^{\frac{\beta}{2}}|x|^{\beta}$.
In this case, by choosing $p=\beta$, Lemma \ref{lem:primitive} is satisfied.
\end{Remark}

\subsection{The Distinction Between Gaussian and Non-Gaussian QMLE}

First of all, consider the case in which $\eta_f$ is given, or more directly, the true error distribution is known. The following asymptotic analysis reveals the difference between the Gaussian QMLE and the non-Gaussian one.
\begin{Theorem}\label{thm:Con1}Assume that $\eta_f$ is known. Under Assumptions \ref{ASS1}, $\hat{\boldsymbol{\theta}}_{\boldsymbol{T}}\stackrel{\rm
\mathcal{P}}{\longrightarrow}{\boldsymbol{\theta_0}}$, where $\hat{\boldsymbol{\theta}}_{\boldsymbol{T}}$ is the quasi likelihood estimator obtained by maximizing (\ref{eq:likelihood}).
\end{Theorem}

Next, we discuss the asymptotic normality of the QMLE. As usual, additional moment conditions are needed.

\begin{Assumption}\label{ASS2}Let $\boldsymbol{k}=({1\over \sigma}, {1\over v_t}{\partial v_t \over \partial \boldsymbol{\gamma}}')'$, and $\boldsymbol{k_0}$ be its value at $\theta = \theta_0$.
\begin{enumerate}
\item $0<E(h^2(\frac{\varepsilon}{\eta_f}))<\infty$, $0<E|\varepsilon \dot{h}(\frac{\varepsilon}{\eta_f})|<\infty$.
\item $\boldsymbol{M}=E(\boldsymbol{k_0}\boldsymbol{k_0}')<\infty$.
\end{enumerate}
\end{Assumption}

\begin{Theorem}\label{thm:CLT1}
Under Assumptions \ref{ASS1} and \ref{ASS2}, we
have
\begin{eqnarray} \label{eq:asym_variance_intheorem}
T^{\frac{1}{2}}(\hat{\boldsymbol{\theta}}_{\boldsymbol{T}}-\boldsymbol{\theta_0})\stackrel{\rm
\mathcal{L}}{\longrightarrow}N\Big(\boldsymbol{0}, \boldsymbol{\Sigma_{1}}
\Big)\end{eqnarray}
where $\boldsymbol{\Sigma_{1}} = \boldsymbol{M}^{-1} \frac{E h_1^2}{(Eh_2)^2}$, $h_1(\ve)=1+h(\frac{\varepsilon}{\eta_f})$, and $h_2(\ve)=\frac{\varepsilon}{\eta_f}\dot{h}(\frac{\varepsilon}{\eta_f})$.
\end{Theorem}

The moment conditions given by the first point of Assumption \ref{ASS2} only depend on the tail of the innovation density $g$ and quasi likelihood $f$. A striking advantage of non Gaussian QMLE over its Gaussian alternative is that the former may require weaker conditions on the tail of the innovation. It is well known that the asymptotic normality of Gaussian likelihood requires a finite fourth moment. By contrast, it implies from Remark \ref{rmk:likelihood} that any Student's $t$ likelihood with degree of freedom larger than 2 has a bounded moment, so that no additional moment conditions are needed other than those assumed in any GARCH model.

Moreover, it turns out that model parameter $\eta_f$ has another interpretation as a bias correction for a simple non-Gaussian QMLE of the scale parameter in that $\sigma_0\eta_f$ would be reported instead of $\sigma_0$. Therefore, the unscaled QMLE can consistently estimate $\sigma_0$ if and only if $\eta_f=1$. Proposition \ref{cor:example} hence reveals the distinction in consistency between the MLE, Gaussian QMLE and the other alternatives.

In general, for an arbitrary likelihood, $\eta_f$ would not equal to 1, thereby creating the popularity of the Gaussian QMLE, whose $\eta_f$ is exactly 1. It is therefore necessary to incorporate this bias-correction factor $\eta_f$ into non-Gaussian QMLE, which may potentially obtain a better efficiency than the Gaussian QMLE. However, as we have no prior information concerning the true innovation density, $\eta_f$ is unknown. As a result, this estimator is infeasible. A promising way to resolve this issue would be to estimate $\eta_f$ in the first step.

\subsection{Two-Step Estimation Procedure}
In order to estimate $\eta_f$, a sample on the true innovation is required. According to Proposition \ref{cor:example}, without knowing $\eta_f$, the residuals from the Gaussian QMLE may potentially provide substitution for the true innovation sample. A two-step estimation procedure is proposed in the following. In the first step, $\hat{\eta}_f$ is obtained by maximizing (\ref{eq:eta_likelihood}) with estimated residuals from Gaussian quasi likelihood estimation:
\begin{equation} \label{eq:ML for etaf}
    \hat \eta_f = \mbox{argmax}_{\eta} {1\over T}\sum_{t=1}^T l_2(x_t, \tilde {\boldsymbol{\theta}}_{\boldsymbol{T}}, \eta) = \mbox{argmax}_{\eta} {1\over T}\sum_{t=1}^T \Big( -\log(\eta) + \log f({\tilde \varepsilon_t \over \eta})  \Big)
    \end{equation}
   where
     \begin{equation}
    \tilde{ \boldsymbol{\theta}}_{\boldsymbol{T}} = \mbox{argmax}_{\boldsymbol{\theta}} {1\over T} \sum_{t=1}^T l_1(x_t,\boldsymbol{\theta}) = \mbox{argmax}_{\boldsymbol{\theta}} {1\over T} \sum_{t=1}^T \Big( -\log (\sigma v_t) - {x_t^2 \over 2 \sigma^2 v_t^2}  \Big)
    \end{equation}
and $\tilde \varepsilon_t = x_t / (\tilde \sigma v_t(\tilde{\boldsymbol{ \gamma}}))  $. Next, we maximize non-Gaussian quasi likelihood with plug-in $\hat \eta_f$ and obtain $\hat{ \boldsymbol{\theta}}_{\boldsymbol{T}}$:
    \begin{equation} \label{eq:ML for 2SNG}
    \hat{ \boldsymbol{\theta}}_{\boldsymbol{T}} = \mbox{argmax}_{\boldsymbol{\theta}}{1\over T} \sum_{t=1}^T l_3(x_t, \hat \eta_f, \boldsymbol{\theta}) = \mbox{argmax}_{\boldsymbol{\theta}} {1\over T} \sum_{t=1}^T \Big( -\log(\sigma v_t) + \log f({ x_t \over \hat \eta_f \sigma v_t})  \Big)
    \end{equation}

We call $\hat{ \boldsymbol{\theta}}_{\boldsymbol{T}}$ the two step non-Gaussian QMLE, 2SNG-QMLE for short. Alternatively, this two-step procedure can be viewed as a one-step generalized methods of moments (GMM) procedure, by considering the score functions. Denote $\tilde {\boldsymbol{s}}(x,\boldsymbol{\theta}, \eta, \boldsymbol{\phi}) =(s_1(x,\boldsymbol{\theta}),
s_2(x, \boldsymbol{\theta}, \eta), s_3(x,\eta, \boldsymbol{\phi}))'$, where
\begin{eqnarray}
s_1(x_t,\boldsymbol{\theta}) &=& {\partial \over \partial \boldsymbol{\theta}} l_1(x_t,\boldsymbol{\theta}) = \boldsymbol{k} \Big(-1 + {x_t^2 \over \sigma^2 v_t^2}\Big)\\
s_2(x_t,\boldsymbol{\theta},\eta) &=& {\partial \over \partial \eta} l_2(x_t, \boldsymbol{\theta},\eta)= -{1\over \eta} \Big( 1+ h({ x_t \over \eta \sigma v_t}) \Big)\\
s_3(x_t,\eta, \boldsymbol{\phi}) &=&{\partial \over \partial \boldsymbol{\phi}} l_3 (x_t,\eta, \boldsymbol{\phi}) = -\boldsymbol{k} \Big( 1+ h({ x_t \over \eta \sigma v_t})  \Big)
\end{eqnarray}
then the estimators are obtained using GMM with identity weighting matrix:
\begin{eqnarray}
(\tilde{\boldsymbol{\theta}}_{\boldsymbol{T}}, \hat \eta_f, \hat{\boldsymbol{\phi}}_{\boldsymbol{T}})=\mbox{argmin}_{\boldsymbol{\theta},\eta,\boldsymbol{\phi}}\frac{1}{T}\sum_{t=1}^T \tilde{\boldsymbol{s}}'(x_t,\boldsymbol{\theta}, \eta, \boldsymbol{\phi})\tilde{\boldsymbol{s}}(x_t,\boldsymbol{\theta}, \eta, \boldsymbol{\phi})
\end{eqnarray}
so our proposed estimator is simply $\hat {\boldsymbol{\theta}}_{\boldsymbol{T}}=\hat{\boldsymbol{\phi}}_{\boldsymbol{T}}$.

\subsection{Asymptotic Theory}
Identification for the parameters $\boldsymbol{\theta}$ and $\eta$ is straightforward. As in Theorem \ref{thm:Con1}, the consistency thereby holds:

\begin{Theorem} \label{thm:Con2}
Given Assumption \ref{ASS1},
$
(\tilde{ \boldsymbol{\theta}}_{\boldsymbol{T}}, \hat \eta_f, \hat{\boldsymbol{ \theta}}_{\boldsymbol{T}})\stackrel{\rm
\mathcal{P}}{\longrightarrow} (\boldsymbol{\theta_0}, \eta_f, \boldsymbol{\theta_0}),
$
in particular the 2SNG-QMLE $\hat{ \boldsymbol{\theta}}_{\boldsymbol{T}}$ is consistent.
\end{Theorem}
In order to obtain the asymptotic normality, we realize that a finite fourth moment for the innovation is essential in that the first step employs the Gaussian QMLE. Although alternative rate efficient estimators may be adopted to avoid moment conditions required in the first step, we prefer the Gaussian QMLE for its simplicity and popularity in practice.

\begin{Theorem}\label{thm:CLT2}
Assume that $E(\varepsilon^4)< \infty$, that Assumptions \ref{ASS1} and \ref{ASS2} are satisfied. Then $(\tilde{\boldsymbol{ \theta}}_{\boldsymbol{T}}, \hat \eta_f, \hat {\boldsymbol{\theta}}_{\boldsymbol{T}})$ are jointly normal asymptotically. That is,
\begin{align*}
&  \left(
\begin{array}
[c]{c}%
T^{1\over 2}(\tilde {\boldsymbol{\theta}}_{\boldsymbol{T}} - \boldsymbol{\theta_0})\\
T^{1\over 2}(\hat \eta_f - \eta_f)\\
T^{1\over 2}(\hat {\boldsymbol{\theta}}_{\boldsymbol{T}} - \boldsymbol{\theta_0})
\end{array}
\right)  \overset{\mathrm{\mathcal{L}}}{\longrightarrow}
N\Big (\left(
\begin{array}
[c]{c}%
\boldsymbol{0}\\
0\\
\boldsymbol{0}
\end{array}
\right)  ,\left(
\begin{array}
[c]{ccc}%
\boldsymbol{\Sigma_{G}} & \boldsymbol{\Pi}' & \boldsymbol{\Xi}\\
\boldsymbol{\Pi} & {\Sigma_{\eta_f}} &\boldsymbol{\Pi}\\
\boldsymbol{ \Xi} & \boldsymbol{\Pi}' &\boldsymbol{\Sigma_{2}}
\end{array}
\right)  \Big)
\end{align*}
where
\begin{eqnarray}
&&\boldsymbol{\Sigma_G}={E(\varepsilon^2 - 1)^2 \over 4} \boldsymbol{M}^{-1}\\
&&\boldsymbol{\Sigma_{2}} = {Eh_1(\ve)^2 \over (Eh_2(\ve))^2}\boldsymbol{M}^{-1} + \sigma_0^2 \Big( {E(\varepsilon^2 - 1)^2 \over 4 } - {Eh_1(\ve)^2 \over (Eh_2(\ve))^2} \Big)\boldsymbol{e_1}\boldsymbol{ e_1}'\label{eq:avar_hat theta}\\
&&{\Sigma_{\eta_f}}=\eta_f^2 E\Big( {\varepsilon^2 - 1\over 2} - {h_1(\ve) \over Eh_2(\ve)} \Big)^2\label{eq:avar_eta}\\
&&\boldsymbol{\Pi}={\eta_f\sigma_0 \over 2}E\Big((\varepsilon^2-1)({h_1(\ve)\over Eh_2(\ve)}-{\varepsilon^2-1\over 2})\Big)\boldsymbol{e'_1}\\
&&\boldsymbol{\Xi}=\frac{E(h_1(\ve)\cdot(\varepsilon^2-1))}{2E(h_2(\ve))}\boldsymbol{M^{-1}}-\frac{\sigma^2_0}{2}E\Big((\varepsilon^2-1)({h_1(\ve)\over Eh_2(\ve)}-{\varepsilon^2-1\over 2})\Big)\boldsymbol{e_1}\boldsymbol{e'_1}
\end{eqnarray}

where $\boldsymbol{e_1}$ is a unit column vector that has the same length as $\boldsymbol{\theta}$, with the first entry one and all the rest zeros.
\end{Theorem}

Before a thorough efficiency analysis of the non-Gaussian QMLE $\hat{\boldsymbol{\theta}}_{\boldsymbol{T}}$, we first discuss the asymptotic property of $\hat \eta_f$. Although $\hat \eta_f$ is obtained using fitted residuals $\tilde \varepsilon_t$ in (\ref{eq:ML for etaf}), the asymptotic variance of $\hat \eta_f$ is not necessarily worse than that using the actual innovations $\varepsilon_t$. In fact, with true innovation the asymptotic variance of the $\eta_f$ estimator is $\eta_f^2 {Eh_1^2/(Eh_2)^2}$. Comparing it with (\ref{eq:avar_eta}), we can find that using fitted residual improves the efficiency as long as the $|{h_1/ Eh_2} - {(\varepsilon^2 -1) / 2}|$ is smaller than $|{h_1 / Eh_2}|$. This occurs when the non-Gaussian likelihood is close to Gaussian likelihood. One extreme case is choosing the same Gaussian likelihood in the second step. Then $\eta_f$ exactly equals one and the asymptotic variance of $\hat \eta_f$ vanishes.

$\eta_f$ also reveals the issue of asymptotic bias incurred by using unscaled non-Gaussian QMLE. From \ref{eq:ML for 2SNG}, while 2SNG-QMLE $\hat{\boldsymbol{\theta}}_{\boldsymbol{T}} = (\hat \sigma_T, \hat{\boldsymbol{\gamma}}_{\boldsymbol{T}})$ maximizes the log-likelihood, unscaled non-Gaussian QMLE would choose estimator $(\hat \eta_f \hat \sigma_T, \hat{\boldsymbol{\gamma}}_{\boldsymbol{T}})$ to maximize log-likelihood without $\eta_f$ in it. So for the volatility scale parameter $\sigma$ it is biased exactly according to the $\hat \eta_f$. Such bias will propagate if using the popular original parametrization. Recall
\begin{eqnarray*}
x_t&=&\sigma_t\varepsilon_t\\
\sigma^2_t&=&\tilde{c}+\sum_{i=1}^p\tilde{a}_i x_{t-i}^2+\sum_{j=1}^q\tilde{b}_j\sigma_{t-j}^2
\end{eqnarray*}
Clearly, we have $\sigma^2a_i=\tilde{a}_i$, $b_j=\tilde{b}_j$ and ${\sigma}^2 = c$. Therefore, potential model misspecification would result in systematic biases in the all estimates of $a_i$ and $c$ if unscaled non-Gaussian MLE, such as Student's $t$-likelihood, is applied without introducing $\eta_f$.

\subsection{Efficiency Gain over Gaussian QMLE}
We compare the efficiency of three estimators of $\boldsymbol{ \theta}$ using two step non-Gaussian QMLE, one step (infeasible) non-Gaussian QMLE with known $\eta_f$, and Gaussian QMLE. Their asymptotic variances are $\boldsymbol{\Sigma_2}$, $\boldsymbol{\Sigma_1}$ and $\boldsymbol{\Sigma_G}$ respectively. The difference in asymptotic variances between the first two estimators is
\begin{eqnarray}
\boldsymbol{\Sigma_{2}} - \boldsymbol{\Sigma_{1}} = \Big(
\begin{array}{cc}
\mu \sigma_0^2 & \boldsymbol{0} \\
\boldsymbol{0} & \boldsymbol{0}
\end{array}\Big)
\end{eqnarray}
where
\begin{eqnarray}
\mu = {E(\varepsilon^2-1)^2 \over 4} - {Eh_1^2 \over (Eh_2)^2}
\end{eqnarray}
Effectively, the sign and magnitude of $\mu$ summarize the advantage of knowing $\eta_f$. $\mu$ is usually positive when the true error has heavy tails while non-Gaussian QMLE is selected to be a heavy-tailed likelihood, illustrating the loss from not knowing $\eta_f$. However, it could also be negative when the true innovation has thin tail, indicating that not knowing $\eta_f$ is actually better when a heavy tail density is selected. Intuitively, this is because the two-step estimator incorporates a more efficient Gaussian QMLE into the estimation procedure. More importantly, the asymptotic variance of $\boldsymbol{\gamma}$ and the covariance between $\sigma$ and $\boldsymbol{\gamma}$ are not affected by the estimation of $\eta_f$. In other words, we achieve the adaptivity property for $\boldsymbol{\gamma}$: with an appropriate non-Gaussian QMLE, $\boldsymbol{\gamma}$ could be estimated without knowing $\eta_f$ equally well as if $\eta_f$ were known before.

We next compare the efficiency between Gaussian QMLE and 2SNG-QMLE. By (\ref{eq:asym_variance_intheorem}) with $f$ replaced by the Gaussian likelihood, we have
\begin{eqnarray*}
\boldsymbol{\Sigma_G}={E(\varepsilon^2 - 1)^2 \over 4} \boldsymbol{M}^{-1}
\end{eqnarray*}
It follows from Lemma \ref{lem:M^-1} in the appendix that,
\begin{eqnarray}\label{eff:QMLEvsMLE}
\boldsymbol{\Sigma_G} - \boldsymbol{\Sigma_{2}} = \mu \Big(
\begin{array}{cc}
\sigma_0^2 \bar{ \boldsymbol{y}}'_{\boldsymbol{0}} \boldsymbol{V} \bar{ \boldsymbol{y}}_{\boldsymbol{0}} & -\sigma_0 \bar {\boldsymbol{y}}_{\boldsymbol{0}}' \boldsymbol{V}\\
-\sigma_0 \boldsymbol{V} \bar {\boldsymbol{y}}_{\boldsymbol{0}} & \boldsymbol{V}
\end{array}\Big)
\end{eqnarray}
where $\boldsymbol{y_0}=v_t(\boldsymbol{\gamma_0})\frac{\partial v_t(\boldsymbol{\gamma_0})}{\partial \boldsymbol{\gamma}}$, $\bar {\boldsymbol{y}}_{\boldsymbol{0}}=E(\boldsymbol{y_0})$, $\boldsymbol{V}=\mbox{Var}(\boldsymbol{y_0})^{-1}$ and hereby the last matrix in (\ref{eff:QMLEvsMLE}) is positive definite.
Therefore as long as $\mu$ is positive,  non-Gaussian QMLE is more efficient for both $\sigma$ and $\boldsymbol{\gamma}$.

It is well known that the financial data sets such as stock prices and exchange rates exhibit heavy tails. Therefore, if a selected likelihood has heavier tails than Gaussian density, then $\mu$ is positive, and the efficiency of the QMLE is thereby improved over Gaussian QMLE.

\subsection{Efficiency Gap from the MLE}

Denote the asymptotic variance of the MLE as $\boldsymbol{\Sigma_M}$. By (\ref{eq:asym_variance_intheorem}) with $f$ replaced by the true likelihood $g$, we have
\begin{eqnarray*}
\boldsymbol{\Sigma_M}=\boldsymbol{M}^{-1}(\int_{-\infty}^{+\infty}x^2\frac{\dot{g}^2}{g}dx-1)^{-1}=\boldsymbol{M}^{-1}(E(h^2_g-1))^{-1}
\end{eqnarray*}
where $h_g=x\frac{\dot{g(x)}}{g(x)}$. The gap in asymptotic variance between 2SNG-QMLE and MLE is given by
\begin{eqnarray*}\label{eq:relative_eff}
\boldsymbol{\Sigma_2}-\boldsymbol{\Sigma_M} = \Big({Eh_1^2 \over (Eh_2)^2}-(E(h^2_g-1))^{-1}\Big)\boldsymbol{M}^{-1} + \boldsymbol{\sigma_0}^2 \Big( {E(\varepsilon^2 - 1)^2 \over 4 } - {Eh_1^2 \over (Eh_2)^2} \Big)\boldsymbol{e_1} \boldsymbol{e_1}'
\end{eqnarray*}

An extreme case is that the selected likelihood $f$ happens to be the true innovation density. Being unaware of it, we still apply a two-step procedure and uses the estimated $\eta_f$. Therefore, the first term in (\ref{eq:relative_eff}) vanishes, but the second term remains. Consequently, $\hat \gamma$ reach the efficiency bounds, while the volatility scale $\hat \sigma$ fails, reflecting the penalty of ignorance of the true model. This example also sheds light on the fact that $\hat{\boldsymbol{\theta}}_{\boldsymbol{T}}$ cannot obtain the efficiency bounds for all parameters unless the true underlying density and the selected likelihood are both Gaussian. This observation agrees with the comparison study in the \cite{GonzlezRivera_Drost1999} concerning the MLE and their semiparametric estimator.

\subsection{The Effect of the First Step Estimation}

We would like to further explore the oracle property of the estimator for heteroscedastic parameters by considering a general first step estimator. We have shown in Theorem \ref{thm:CLT2} that the efficiency of the estimator for $\boldsymbol{\gamma}$ is not affected by the first step estimation of $\eta_f$ using Gaussian QMLE, as if $\eta_f$ were known. Therefore, we may relax the finite fourth moment requirement on the innovation error by applying another efficient estimator in the first step. On the other hand, even if the first step estimator suffers from a lower rate, it may not affect the efficiency of the heteroscedastic parameters $\boldsymbol{\gamma}$, which is always $T^{1\over2}$ consistent and asymptotically normal.

\begin{Theorem}\label{thm:CLT3} Suppose that the first step estimator $\tilde{\boldsymbol{\theta}}$ has an influence function representation:
\begin{eqnarray*}
T \lambda_T^{-1}(\tilde{\boldsymbol{\theta}}-\boldsymbol{\theta_0})=\lambda_T^{-1}\sum_{t=1}^T\boldsymbol{\Psi_t}(\varepsilon_t)+\boldsymbol{o_P(1)}
\end{eqnarray*}
with the right hand side converging to a non-degenerate distribution, and $\lambda_T\sim T^{1/\alpha}$ for some $\alpha\in[1,2]$. Then the convergence rate for $\sigma$ is also $T \lambda_T^{-1}$, while
the same central limit theorem for $\gamma$ as in Theorem \ref{thm:Con2} remains, that is,
\begin{eqnarray*}
T^{\frac{1}{2}}(\boldsymbol{\gamma}-\boldsymbol{\gamma_0})\overset{\mathrm{\mathcal{L}}}{\longrightarrow}N(\boldsymbol{0},{Eh_1^2 \over (Eh_2)^2} \boldsymbol{V})
\end{eqnarray*}
where $\boldsymbol{V}=\Big(\mbox{Var}(\frac{1}{\nu(\boldsymbol{\gamma_0})}\frac{\partial \nu}{\partial \boldsymbol{\gamma}}|_{\boldsymbol{\gamma}=\boldsymbol{\gamma_0}})\Big)^{-1}$.
\end{Theorem}

Theorem \ref{thm:CLT3} applies to several estimators that have been discussed in the literature. For example, \cite{HY2003} have discussed the Gaussian QMLE with ultra heavy-tailed innovations that violate a finite fourth moment. In their analysis, $\lambda_T$ is regularly varying at infinity with exponent $\alpha\in [1,2)$. The resulting Gaussian QMLE $\tilde{\boldsymbol{\theta}}$ suffers lower convergence rates. By contrast, \cite{DrostKlaassen1997} have suggested an M-estimator based on the score function for logistic distribution to avoid moment conditions on the innovations. Both estimators, if applied in the first step, would not affect the efficiency of $\hat{\boldsymbol{\gamma}}_{\boldsymbol{T}}$.

\section{Extensions}
We discuss two ways to further improve the efficiency of 2SNG-QMLE. One is choosing the non-Gaussian likelihood from a pool of candidate likelihoods to adapt to data, the other is an affine combination of 2SNG-QMLE and Gaussian QMLE according their covariance matrix in Theorem \ref{thm:CLT2} to minimize resulting estimator's asymptotic variance.

\subsection{Optimal Choice of Likelihood}

There are two distinctive edges of choosing a heavy tailed quasi likelihood over Gaussian
likelihood. First, the $T^{1\over2}$-consistency of 2SNG-QMLE of $\boldsymbol{\gamma}$ no
longer depends on finite fourth moment condition, but instead finite $Eh_1^2/(Eh_2)^2$. This can be easily met by, for example, choosing generalized Gaussian likelihood with $\beta \leq 1$.
Second, even under finite fourth moment, heavy tailed 2SNG-QMLE has lower variance than Gaussian
QMLE if true innovation is heavy tailed. A pre-specified heavy tailed likelihood can have these
two advantages. However, we can adaptively choose this quasi likelihood to further improve its
efficiency. This is done by minimizing the efficiency loss from MLE, which is equivalent by minimizing $Eh_1^2/(Eh_2)^2$ over certain families of heavy tailed likelihoods. We propose optimal choice of non-Gaussian
likelihoods, where candidate likelihoods are from Student's $t$ family with degree of freedom $\nu>2$ and generalized Gaussian family with $\beta \leq 1$. Formally, for true innovation distribution $g$ and candidate
likelihood $f$, define
\begin{equation}
A(f,g) = {E_g h_1^2 \over E_g(h_2)^2}, \quad \mbox{where } h_1 = 1+ h(\frac{\varepsilon}{\eta_f}),\mbox{ and
} h_2 = \frac{\varepsilon}{\eta_f}\dot{h}(\frac{\varepsilon}{\eta_f})
\end{equation}
Then the optimal likelihood is chosen from $t$-family and generalized Gaussian family (gg):
\begin{equation} \label{eq:choose likelihood}
f^* = \mbox{argmin}_{\nu, \beta}\left\{ \{ A(f_{\nu}^t,\hat g)\}_{\nu >2},  \{
A(f_{\beta}^{gg},\hat g)\}_{\beta \leq 1}  \right\}
\end{equation}
where $\hat g$ denotes the empirical distribution of estimated residuals from Gaussian QMLE, the
first step. Because this procedure of choosing likelihood is adaptive to data, it is expected that
the chosen quasi likelihood results in a more efficient 2SNG-QMLE than a pre-specified one. We
justify this point in simulation studies.

A 2SNG-QMLE with choosing optimal likelihood runs the following four steps: (a) Run Gaussian QMLE
and obtain the estimated residuals; (b) Run optimization (\ref{eq:choose likelihood}) and obtain
the optimal likelihood $f^*$; (c) Obtain $\hat \eta_f$ using $f^*$ and estimated residuals; (d) Run
2SNG-QMLE with $f^*$ and $\hat \eta_f$.

\subsection{Aggregating 2SNG-QMLE and Gaussian QMLE}
Another way to further improve the efficiency of 2SNG-QMLE is through aggregation. Since both Gaussian QMLE and 2SNG-QMLE are consistent, an affine combination of the two, with weights chosen according to their joint asymptotic variance, yields a consistent estimator and is more efficient than both. Define the aggregation estimator
\begin{eqnarray}
\hat \Btheta _{\BT} ^ {\BW} = \BW \hat \Btheta + (\boldsymbol{I} - \BW) \tilde \Btheta
\end{eqnarray}
where $\BW$ is a diagonal matrix with weights $(w_1, w_2,\ldots, w_{1+p+q})$ on the diagonal. From Theorem \ref{thm:CLT2}, the optimal weights are chosen from minimizing the asymptotic variance of each component of the aggregation estimator:
\begin{eqnarray} \label{eq:aggre_general_weights}
w_i^* &=& \mbox{argmin}_w w^2 (\boldsymbol{\Sigma_2})_{i,i} + (1-w)^2 (\boldsymbol{\Sigma_G})_{i,i} + 2 w (1-w) \boldsymbol{\Xi}_{i,i} \\ \nonumber
      &=& {(\boldsymbol{\Sigma_G})_{i,i} - \boldsymbol{\Xi}_{i,i} \over (\boldsymbol{\Sigma_2})_{i,i} + (\boldsymbol{\Sigma_G})_{i,i} - 2 \boldsymbol{\Xi}_{i,i} }
\end{eqnarray}
It turns out that all optimal aggregation weights $w_i^*$ are the same, which is
\begin{eqnarray} \label{eq:aggre_weights}
w^* = {E\big({1 - \ve^2  \over 2}({1 - \ve^2  \over 2} + {h_1 \over Eh_2})\big)   \over E \big({1 - \ve^2  \over 2} + {h_1 \over Eh_2}\big)^2  }.
\end{eqnarray}
\begin{Proposition} \label{Prop:aggregation}
The aggregated estimator $\hat \Btheta_{\BT}^*$ uses optimal aggregation weights $\BW^* = w^*\boldsymbol{I}$. Its asymptotic variance has diagonal terms
\begin{eqnarray}
\boldsymbol{\Sigma}^*_{i,i} = { (\boldsymbol{\Sigma_2})_{i,i}(\boldsymbol{\Sigma_G})_{i,i} - \boldsymbol{\Xi}_{i,i}^2 \over (\boldsymbol{\Sigma_2})_{i,i} + (\boldsymbol{\Sigma_G})_{i,i} - 2 \boldsymbol{\Xi}_{i,i} }, \quad i = 1,\ldots,1+p+q.
\end{eqnarray}
\end{Proposition}
Although estimators for $\sigma$ and $\gamma$ have different asymptotic properties, the optimal aggregation weights are the same: $w_i^* = w^*$. Also the weight depends only on non-Gaussian likelihood and innovation distribution, but not on GARCH model specification. The aggregated estimator $\hat \Btheta_{\BT}^*$ always have smaller asymptotic variance than both 2SNG-QMLE and Gaussian QMLE. If data is heavy tailed, e.g., $E\ve^4$ is large or equal to $\infty$, from (\ref{eq:aggre_weights}) it simply assigns weights approximately 1 for 2SNG-QMLE and 0 for Gaussian QMLE. In practice, after running 2SNG-QMLE with optimal choice of likelihood, we can estimate the optimal aggregation weight $w^*$ by plugging into (\ref{eq:aggre_weights}) the estimated residuals.

\section{Simulation Studies}
\subsection{Model Free Characteristics}
The scale tuning parameter $\eta_f$ and the efficiency difference $\mu$ are generic characteristics of non-Gaussian likelihoods and of the true innovations, and they do not change when using another conditional heteroscedastic model. We numerically evaluate how they vary according to the non-Gaussian likelihoods and innovations.

\begin{table}[h!]
\begin{center}
\caption{$\eta_f$ for generalized Gaussian QMLEs ($gg$,row) and innovation distributions (column)} \label{tab:eta_f table GGD}
\begin{tabular}{rrrrrrr|rrrr}
\hline
\hline
& $gg_{0.2}$ & $gg_{0.6}$ & $gg_{1}$ & $gg_{1.4}$ & $gg_{1.8}$ & $gg_2$ & $t_3$ & $t_5$ & $t_7$ & $t_{11}$ \\
 \hline
 $gg_{0.2}$ & 1.000 & 6.237 & 8.901 & 10.299 & 11.125 & 11.416 & 8.128 & 9.963 & 10.483 & 10.885 \\
 $gg_{0.6}$ & 0.271 & 1.000 & 1.291 & 1.434 & 1.515 & 1.544 & 1.159 & 1.384 & 1.443 & 1.487 \\
 $gg_{1.0}$ & 0.354 & 0.844 & 1.000 & 1.073 & 1.114 & 1.128 & 0.900 & 1.040 & 1.074 & 1.098 \\
 $gg_{1.4}$ & 0.537 & 0.873 & 0.962 & 1.000 & 1.022 & 1.029 & 0.883 & 0.977 & 0.998 & 1.012 \\
 $gg_{1.8}$ & 0.811 & 0.952 & 0.981 & 0.993 & 1.000 & 1.002 & 0.946 & 0.985 & 0.991 & 0.997 \\
 \hline
 \end{tabular}
\end{center}
\end{table}

\begin{table}[h!]
\begin{center}
\caption{$\eta_f$ for Student's t QMLEs (row) and innovation distributions (column)} \label{tab:eta_f table Student_t}
\begin{tabular}{rrrrrrr|rrrr}
\hline
\hline
 & $t_{2.5}$ & $t_3$ & $t_4$ & $t_5$ & $t_7$ & $t_{11}$   & $gg_{0.5}$ & $gg_{1}$ & $gg_{1.5}$ & $gg_2$ \\
 \hline
$t_{2.5}$ & 1.000 & 1.231 & 1.425 & 1.506 & 1.584 & 1.641 & 0.900 & 1.414 & 1.614 & 1.716 \\
$t_{3}$ &   0.815 & 1.000 & 1.151 & 1.216 & 1.275 & 1.318 & 0.756 & 1.150 & 1.301 & 1.375 \\
$t_{4}$ &   0.715 & 0.874 & 1.000 & 1.054 & 1.100 & 1.133 & 0.697 & 1.011 & 1.122 & 1.174 \\
$t_{5}$ &   0.690 & 0.836 & 0.953 & 1.000 & 1.043 & 1.071 & 0.691 & 0.966 & 1.061 & 1.107 \\
$t_{7}$ &   0.679 & 0.816 & 0.922 & 0.964 & 1.000 & 1.024 & 0.708 & 0.945 & 1.018 & 1.053 \\
$t_{11}$  & 0.690 & 0.823 & 0.916 & 0.953 & 0.980 & 1.000 & 0.749 & 0.941 & 0.998 & 1.021 \\
$t_{20}$  & 0.720 & 0.845 & 0.928 & 0.958 & 0.981 & 0.992 & 0.811 & 0.954 & 0.992 & 1.007 \\
$t_{30}$  & 0.742 & 0.862 & 0.939 & 0.965 & 0.981 & 0.992 & 0.846 & 0.966 & 0.993 & 1.004 \\
 \hline
 \end{tabular}
\end{center}
\end{table}

Table \ref{tab:eta_f table GGD} and \ref{tab:eta_f table Student_t} show how $\eta_f$ varies over generalized Gaussian likelihoods and Student's $t$ likelihoods with different parameters respectively. For each row, which amounts to fixing a quasi likelihood, the lighter the tails of innovation errors are, the larger $\eta_f$. Furthermore $\eta_f>1$ for innovation errors that are lighter than the likelihood, and $\eta_f<1$ for innovations that are heavier than the likelihood. Therefore if the non-Gaussian likelihood have heavier tails than true innovation, we should shrink the data in order for consistent estimation. On the other hand if the quasi likelihood is lighter than true innovation, we should magnify the data.

For each column (fix an innovation distribution), in most cases the heavier the tails of likelihoods are, the larger $\eta_f$, but the monotone relationship is not true for some ultra heavy tail innovations, in which cases $\eta_f$ shows a ``smile" dynamic. The non-monotonicity in the likelihood dimension indicates that to determine $\eta_f$ one needs more information about the likelihood than just the asymptotic behavior of its tails.

 Table \ref{tab:mu table GGD} and \ref{tab:mu table Student_t} show the dependence of $\mu$ on the true innovation (column) and non-Gaussian likelihood (row). From the table we see that in most cases $\mu$ is positive, which means that non-Gaussian QMLE shows an improvement. But when heavy tailed likelihoods are applied on true innovations with moderate or thin tails, $\mu$ turns negative, which means that Gaussian QMLE performs better.

Looking at each column, by fixing the innovation distribution, non-Gaussian QMLE performs the best over Gaussian QMLE when the non-Gaussian likelihood coincides with the innovation distribution (MLE). Looking at each row, by fixing a non-Gaussian likelihood, its relative performance increases when the true innovation distributions become more heavy tailed, even after passing the MLE point where true innovation and likelihood are the same. This is because $\mu$ is a relative measure of non-Gaussian over Gaussian, not an absolute measure for asymptotic variance. When the true innovation is heavier than the non-Gaussian likelihood, non-Gaussian QMLE does not perform as well as MLE, but Gaussian QMLE does even worse than as if the true innovation coincides with non-Gaussian likelihood. Therefore, even the absolute efficiency in terms of asymptotic variance drops for non-Gaussian QMLE, its relative performance over Gaussian QMLE actually increases.

To summarize the variation of $\mu$, one can draw a line for distributions according to their asymptotic behavior of tails, in other words, according to how heavy their tails are, with thin tails on the left and heavy tails on the right. Then we place non-Gaussian likelihood, Gaussian likelihood and true innovation distribution onto this line. The sign and value of $\mu$ depends on where true innovation distribution is placed. (a) It is placed on the right side of non-Gaussian likelihood, then $\mu$ is positive and large. (b) Error is on the left side of Gaussian, then $\mu$ is negative and large in absolute value. (c) Error is between non-Gaussian and Gaussian, then, to which likelihood is innovation closer determines $\mu$. This seems like a symmetric argument for Gaussian and non-Gaussian likelihood. But in financial applications we know true innovations are heavy tailed. Even the non-Gaussian likelihood may not be the innovation distribution, we still can guarantee either (a) happens or (c) happens with innovation closer to non-Gaussian likelihood. In both cases we have $\mu>0$ and non-Gaussian QMLE is a more efficient procedure than Gaussian QMLE.

\begin{table}[h!]
\begin{center}
\caption{$\mu$ for generalized Gaussian QMLEs ($gg$,row) and innovation distributions (column)} \label{tab:mu table GGD}
\begin{tabular}{rrrrrrr|rrrr}
\hline
\hline
 & $gg_{0.2}$ & $gg_{0.6}$ & $gg_{1}$ & $gg_{1.4}$ & $gg_{1.8}$ & $gg_2$ & $t_{4.5}$ & $t_5$ & $t_7$ & $t_{11}$ \\
 \hline
 $gg_{0.2}$ & 484 & 1.773 &-0.062 &-0.335 &-0.416 &-0.436 & 2.411 & 0.929 &-0.026 &-0.274 \\
 $gg_{0.6}$ & 482 & 1.978 & 0.195 &-0.075 &-0.157 &-0.178 & 2.608 & 1.138 & 0.206 &-0.030 \\
 $gg_{1.0}$ & 474 & 1.839 & 0.250 & 0.017 &-0.053 &-0.071 & 2.590 & 1.149 & 0.267 & 0.054 \\
 $gg_{1.4}$ & 443 & 1.424 & 0.209 & 0.040 &-0.010 &-0.022 & 2.369 & 1.008 & 0.234 & 0.068 \\
 $gg_{1.8}$ & 328 & 0.589 & 0.089 & 0.022 & 0.003 &-0.002 & 1.588 & 0.596 & 0.114 & 0.032 \\
 \hline
 \end{tabular}
\end{center}
\end{table}

\begin{table}[h!]
\begin{center}
\caption{$\mu$ for Student's t QMLEs (row) and innovation distributions (column)} \label{tab:mu table Student_t}
\begin{tabular}{rrrrrrr|rrrr}
\hline
\hline
 & $t_{4.5}$ & $t_5$ & $t_7$ & $t_9$ & $t_{15}$ & $t_{30}$   & $gg_{0.5}$ & $gg_{1}$ & $gg_{1.5}$ & $gg_2$ \\
 \hline
$t_{2.5}$ & 2.534 & 1.045 & 0.071 & -0.114 & -0.263 & -0.324 & 3.848 & 0.004 & -0.296 & -0.375 \\
$t_{3}$ &   2.626 & 1.145 & 0.189 & 0.011 &-0.124 &-0.183 & 3.871 & 0.124 & -0.158 & -0.223 \\
$t_{4}$ &   2.663 & 1.194 & 0.258 & 0.086 &-0.038 &-0.090 & 3.816 & 0.191 &-0.067 &-0.124 \\
$t_{5}$ &   2.664 & 1.200 & 0.277 & 0.114 &-0.004 &-0.054 & 3.770 & 0.211 &-0.031 &-0.084 \\
$t_{7}$ &   2.642 & 1.190 & 0.287 & 0.131 & 0.020 &-0.022 & 3.667 & 0.222 &-0.001 &-0.051 \\
$t_{11}$  & 2.591 & 1.150 & 0.277 & 0.132 & 0.035 &-0.004 & 3.500 & 0.212 & 0.016 &-0.025 \\
 \hline
 \end{tabular}
\end{center}
\end{table}

\subsection{Verification of the Asymptotic Theory}
Now we verify the asymptotic formula (\ref{eq:avar_hat theta})-(\ref{eq:avar_eta}). We run $N = 20000$ simulations, each generating a sample of size $T = 7000$ from a GARCH$(1,1)$ model. The model parameters are $\sigma_0 = 0.5$, $a_{10} = 0.35$, $b_{10} = 0.3$. The innovation errors are standardized skewed Student's $t$ distribution with degree of freedom $\nu_0 = 7$ and degree of skewness $\lambda_0 = 0.5$, so that the left tail is heavier than the right tail. We use Student's $t$ likelihood with degree of freedom $\nu = 4$ in non-Gaussian QMLE. We run two-step procedure to obtain the estimates $\hat \eta_f$ and non-Gaussian QMLE estimates $\hat \theta$. Figure \ref{figure:GARCH model verification} reports the standardized estimates of $(\sigma_0, a_{10}, b_{10}, \eta_f)$ compared to $N(0,1)$. Standardization is done by first subtracting the estimates by the true value, and then dividing by the theoretical asymptotic standard deviation according to Theorem \ref{thm:CLT2}. All plots confirm the validity of asymptotic variance formula (\ref{eq:avar_hat theta})-(\ref{eq:avar_eta}).

\begin{figure}[!h]
\begin{center}
   \includegraphics[scale=0.8]{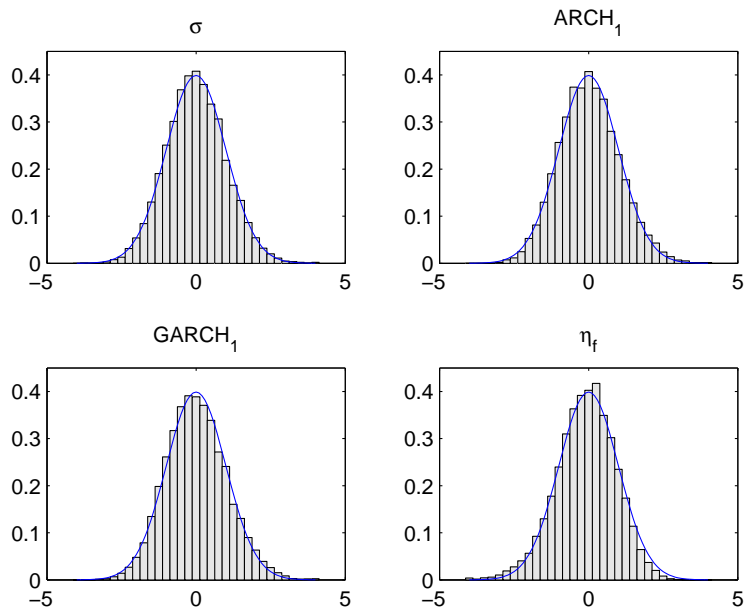}
     \caption{\small \textnormal{Histogram of standardized 2SNG-QMLE and standard normal pdf (solid line). Normalization is done by first subtracting the estimates by the true value, and then dividing by the theoretical asymptotic standard deviation suggested by our theory.}}
   \label{figure:GARCH model verification}
\hfill
\end{center}
\end{figure}

\subsection{Comparison with Gaussian QMLE and MLE}
We compare the efficiency of 2SNG-QMLE, Gaussian QMLE and MLE under various innovation error distributions. We don't perform optimal choice of quasi likelihood in 2SNG-QMLE, instead fix the quasi likelihood to be Student's $t$ distribution with degree of freedom 4. The simulation is conducted on a GARCH$(1,1)$ model with true parameters $(\sigma_0,a_{1,0},b_{1,0})=(0.5,0.35,0.3)$. For innovation errors we use Student's $t$ and generalized Gaussian distributions of various degrees of freedoms to generate data. For each type of innovation distribution, we run $N=1000$ simulations each with $T=3000$ samples. Tables \ref{tab:efficiency_compare_wMLE_t} and \ref{tab:efficiency_compare_wMLE_gg} reports the relative efficiencies of these three estimators in terms of ratios of sample variances and MSEs. The first ratio, Gaussian/2SNG, indicates how 2SNG-QMLE outperforms (underperforms) Gaussian QMLE. The second ratio 2SNG/MLE indicates how far 2SNG-QMLE is from efficiency bound.
\begin{table}[h!]
\begin{center}
 \caption{Student's t innovations simulation}\label{tab:efficiency_compare_wMLE_t}
\begin{tabular}{rrrrrrrr}
\hline
\hline
Innov. & Comparing & \multicolumn{3}{l}{Ratio of variances}  & \multicolumn{3}{l}{Ratio of MSEs} \\
dist. & methods & $\sigma_0$ & $a_{1,0}$ & $b_{1,0}$ & $\sigma_0$ & $a_{1,0}$ & $b_{1,0}$ \\
 \hline
 $t_{20}$ & G./2SNG  &   0.929  &  0.901   &  0.936  &  0.929  &  0.898 &   0.936 \\\vspace{0.12cm}
          & 2SNG/MLE &  1.092  & 1.122    & 1.089   & 1.091   & 1.126  &  1.089 \\
 $t_{15}$ & G./2SNG  &   0.942  &  0.960  &  0.961  &  0.939  &  0.948  &  0.960 \\        \vspace{0.12cm}
          & 2SNG/MLE&   1.112  &  1.121  &  1.087  &  1.114  &  1.131  &  1.087 \\ \vspace{0.12cm}

  $t_{9}$ & G./2SNG  &    1.115&    1.186&    1.108&    1.118&    1.185&    1.109\\        \vspace{0.12cm}
          & 2SNG/MLE&    1.109&    1.022&    1.020&    1.019&    1.023&    1.020\\
  $t_{7}$ & G./2SNG  &    1.216&    1.260&    1.186&    1.217&    1.266&    1.186\\        \vspace{0.12cm}
           & 2SNG/MLE&   1.036&    1.024&    1.031&    1.037&    1.026&    1.031\\

  $t_{6}$ & G./2SNG  &     1.355&    1.528&    1.302&    1.355&    1.552&    1.303\\        \vspace{0.12cm}
          & 2SNG/MLE&   1.&    1.022&    1.&    1.&    1.022&    1.\\

  $t_{5}$ & G./2SNG  &    1.526&    2.495&    1.405&    1.547&    2.530&    1.409\\        \vspace{0.12cm}
        & 2SNG/MLE&   1.025&    1.001&    1.015&    1.027&    1.001&    1.015\\

  $t_{4}$ & G./2SNG  &    2.074&    7.244&    1.847&    2.125&    7.478&    1.858\\        \vspace{0.12cm}
        & 2SNG/MLE&    1.065&   1.&    1.&    1.071&    1.&    1.\\

  $t_{3}$ & G./2SNG  &    2.687&    31.40&    2.535&    2.850&    33.26&    2.580\\
        & 2SNG/MLE&     1.235&    1.&    1.&    1.264&    1.&    1.\\
  $t_{2.5}$ & G./2SNG  &    1.960&    93.91&    2.649&    2.051&    101.5&    2.664\\        \vspace{0.12cm}
        & 2SNG/MLE&     2.371&    1.037&    1.062&    2.625&    1.037&    1.062\\

\hline
 \end{tabular}

\end{center}
\end{table}

In Table \ref{tab:efficiency_compare_wMLE_t} the innovation distributions range from thin-tailed $t_{20}$ (approximately Gaussian) to heavy-tailed $t_{2.5}$. Biases are small so standard deviations and RMSEs are nearly the same. For the first two thin-tailed cases, $t_{20}$ and $t_{15}$, Gaussian QMLE outperforms 2SNG-QMLE by a small margin. For all other cases 2SNG-QMLE outperforms Gaussian QMLE. In heavy tailed cases $t_6$ and $t_5$, 2SNG-QMLE performs nearly as well as MLE, and reduces standard deviations by 15\% to 60\% from Gaussian QMLE. In ultra-heavy tail cases ($t_4$, $t_3$ and $t_{2.5}$), since fourth moment no longer exists, Gaussian QMLE is not $T^{1\over2}$-consistent, and its estimation precision quickly deteriorates, sometimes to an intolerable level. In contrast 2SNQ-QMLE using $t_4$ likelihood does not require finite fourth moment for $T^{1\over2}$-consistent $a_{1,0}$ and $b_{1,0}$, so standard deviations for $a_{1,0}$ and $b_{1,0}$ are still nearly equal to MLE. Standard deviations of $\sigma_0$ are now larger than MLE, but still significantly smaller than Gaussian QMLE.

\begin{table}[h!]
\begin{center}
 \caption{generalized Gaussian innovations simulation}\label{tab:efficiency_compare_wMLE_gg}
\begin{tabular}{rrrrrrrr}
\hline
\hline
Innov. & Comparing & \multicolumn{3}{l}{Ratio of variances}  & \multicolumn{3}{l}{Ratio of MSEs} \\
dist. & methods & $\sigma_0$ & $a_{1,0}$ & $b_{1,0}$ & $\sigma_0$ & $a_{1,0}$ & $b_{1,0}$ \\
 \hline
 $gg_{4}$ & G./2SNG  &   0.743  &  0.742   &  0.769  &  0.748  &  0.736 &   0.771\\
 \vspace{0.12cm}
          & 2SNG/MLE &  1.705  & 1.843    & 1.571   & 1.696   & 1.886  &  1.566 \\
 $Gauss.$ & G./2SNG  &   0.811  &  0.717  &  0.850  &  0.808  &  0.706  &  0.850 \\        \vspace{0.12cm}
          & 2SNG/MLE&   1.233  &  1.395  &  1.176  &  1.238  &  1.416  &  1.176 \\ \vspace{0.12cm}

  $gg_{1.2}$ & G./2SNG  &    1.045&    1.007&    1.019&    1.047&    1.006&    1.016\\        \vspace{0.12cm}
          & 2SNG/MLE&    1.076&    1.113&    1.070&    1.076&    1.117&    1.071\\
  $gg_{1}$ & G./2SNG  &    1.091&    1.210&    1.073&    1.090&    1.201&    1.073\\        \vspace{0.12cm}
           & 2SNG/MLE&   1.084&    1.120&    1.074&    1.086&    1.130&    1.074\\

  $gg_{0.8}$ & G./2SNG  &     1.258&    1.736&    1.237&    1.239&    1.689&    1.229\\        \vspace{0.12cm}
          & 2SNG/MLE&   1.082.&    1.022&    1.044&    1.096&    1.068&    1.048\\

  $gg_{0.6}$ & G./2SNG  &    1.653&    2.623&    1.526&    1.663&    2.650&    1.527\\        \vspace{0.12cm}
        & 2SNG/MLE&   1.089&    1.135&    1.061&    1.100&    1.144&    1.061\\

  $gg_{0.4}$ & G./2SNG  &    1.951&    4.619&    1.772&    1.958&    4.760&    1.764\\        \vspace{0.12cm}
        & 2SNG/MLE&    1.170&   1.204&    1.095&    1.191&    1.210&    1.098\\

\hline
 \end{tabular}

\end{center}
\end{table}

In Table \ref{tab:efficiency_compare_wMLE_gg}, the innovations innovations range from thin tailed $gg_4$ to heavy tailed $gg_{0.4}$. For innovation with $gg_{1.2}$ and heavier, 2SNG-QMLE starts to outperform Gaussian QMLE. In all cases, the Student $t_4$ 2SNG-QMLE performs very close to MLE as indicated by standard deviations. In comparison, Gaussian QMLE's performance deteriorates as tails grow heavier, particulary in $gg_{0.6}$ and $gg_{0.4}$, although in these cases the fourth moments are finite.

\subsection{Ultra-Heavy Tail Case}
Here we compare the efficiency when innovation are transformations from stable-$\alpha$ distributions. Index $\alpha$ ranges from 1.9 down to 1.1, and the distributions are transformed such that they do not have fourth moments but have 2nd moment to be unity. Furthermore they are asymmetric and not unimodal. Since distribution functions are not explicit, MLE is difficult to obtain. Table \ref{tab:efficiency_compare_stable_opt} compares the performance between 2SNG-QMLE with optimally chosen quasi likelihood and Gaussian QMLE. We still use the GARCH(1,1) model with true parameters $(\sigma_0,a_{1,0},b_{1,0})=(0.5,0.35,0.3)$, and run $N=2000$ simulation with $T=3000$. The candidate quasi likelihoods are Student's $t$ distributions with DoF from $20$ to $2.5$, and generalized Gaussian distributions with DoF from $4$ to $0.4$.

\begin{table}[h!]
\begin{center}
 \caption{Stable innovations simulation}\label{tab:efficiency_compare_stable_opt}
\begin{tabular}{rrrrrrrr}
\hline
\hline
Innov. & Comparing & \multicolumn{3}{l}{Ratio of variances}  & \multicolumn{3}{l}{Ratio of MSEs} \\
dist. & methods & $\sigma_0$ & $a_{1,0}$ & $b_{1,0}$ & $\sigma_0$ & $a_{1,0}$ & $b_{1,0}$ \\
 \hline \vspace{0.12cm}
$\alpha=1.9$ & G./NG-opt &    1.266&    1.446&    1.205&    1.285&    1.470&    1.215\\ \vspace{0.12cm}

$\alpha=1.7$ & G./NG-opt &    2.502&    5.072&    2.175&    2.551&    5.301&    2.178\\ \vspace{0.12cm}

$\alpha=1.5$ & G./NG-opt &    5.381&    148.9&    4.004&    5.605&    154.8&    3.954\\ \vspace{0.12cm}

$\alpha=1.3$ & G./NG-opt &    9.774&    499.1&    6.911&    10.10&    524.5&    6.868\\ \vspace{0.12cm}

$\alpha=1.1$ & G./NG-opt &    16.08&    1313&    10.19&    16.94&    1445&    9.960\\
\hline
 \end{tabular}

\end{center}
\end{table}

Gaussian QMLE deteriorates as tails grow heavier (smaller $\alpha$). In particular for $a_{1,0}$, it produces many large estimates, making substantial biases upward and intolerable standard deviation levels. In contrast, 2SNG-QMLE shows little sample bias and small standard deviations. It also shows that as innovations grows heavier, 2SNG-QMLE delivers smaller standard deviations.

For $\alpha=1.9$ case, among 2000 simulations, the algorithm chooses Student's t quasi likelihoods for 1977 times, and Gaussian likelihoods 23 times. Among the chosen Student's t likelihoods, the degrees of freedom spread out from 5 to 20, and mostly concentrate on 6, 7, 9 and 12. For the rest four cases, all chosen quasi likelihoods are Student's t likelihoods. In case $\alpha = 1.7$, the degrees of freedom concentrate on 3 and 4, with a small fraction of 5. In $\alpha=1.5$, around 1650 simulations choose $t_{2.5}$, the rest choose $t_3$. In $\alpha = 1.3$ and $\alpha = 1.1$, all chosen quasi likelihoods are $t_{2.5}$, the most heavy tailed candidate.

\section{Empirical Work}
Work run a simple GARCH$(1,1)$ model on Citigroup stock daily return from January 03, 2008 to Jan 15, 2010. There are 514 trading days in the data. We report the estimated parameters using old parametrization. The Gaussian QMLE estimates for $(c,\tilde a, \tilde b)$ is $(0.6522, 0.2205, 0.7793)$. Clearly data shows high degree of persistence in that $\tilde a +\tilde b \approx 1$. The 2SNG-QMLE chooses $gg_{1.2}$ as optimal likelihood, and the estimates for model parameters and $\eta_f$ are $(0.7689, 0.2075, 0.7728)$ and 1.0458, respectively. Since $\hat \eta_f$ deviates from 1 about 4.6\%, there would be a significant bias if we run $gg_{1.2}$ QMLE without scale adjustment.

On the other hand, even a non-Gaussian QMLE allowing to estimate shape of quasi likelihood cannot guarantee consistency. In fact, such method only picks one likelihood in some distribution family that is \``least" biased for the data, but bias due to misspecification of innovation distribution remains. We perform unscaled generalized Gaussian QMLE with shape estimation. The estimated shape is $\hat \beta = 1.305$, which is close to the 1.2, the shape of optimal likelihood in 2SNG-QMLE. We fix shape 1.305 and run 2SNG-QMLE again, $\hat \eta_f$ is still 1.033. This means even allowing to estimate the shape in quasi-likelihood, unscaled non-Gaussian QMLE still incurs a 3.3\% bias.

\section{Conclusion}
This paper regards on GARCH model estimation when innovation distribution is unknown, and it questions the efficiency issue of Gaussian QMLE and consistency issue of currently used non-Gaussian QMLE. It proposed the 2SNG-QMLE to tackle both issues. The first step runs a Gaussian QMLE whose purpose is to identify the scale tuning parameter, $\eta_f$. The second step runs a non-Gaussian QMLE to estimate model parameters. The quasi likelihood $f$ used in second step can be a pre-specified heavy tailed likelihoods, properly scaled by $\eta_f$. It can also be chosen from a pool of candidate distributions in order to adapt different characteristics of unknown innovation distribution.

The asymptotic theory of 2SNG-QMLE does not depend on any symmetric or unimodal assumptions of innovations. By adopting a different parametrization proposed by \cite{Newey1997}, and incorporating $\eta_f$, 2SNG-QMLE improves the estimation efficiency from Gaussian QMLE. We and show that the asymptotic behavior of 2SNG-QMLE can be broken down to two parts. For the heteroscedastic parameters $\gamma$, 2SNG-QMLE is always $T^{1\over2}$-consistent and asymptotically normal, whereas $T^{1\over2}$ consistency of Gaussian QMLE relies on finite fourth moment assumption. When $E\varepsilon_t^4 < \infty$, 2SNG-QMLE outperforms Gaussian QMLE in term of smaller asymptotic variance, provided that innovation distribution is reasonably heavy tailed, which is common for financial data. For the scale part $\sigma$, 2SNG-QMLE is not always $T^{1\over2}$-consistent, but simulation shows that the estimation for $\sigma$ is usually equally well as heteroscedastic parameters, $\gamma$. We also run simulation to compare the performance of Gaussian QMLE, 2SNG-QMLE and MLE. In most cases 2SNG-QMLE shows an edge and is close to MLE.

One possible generalization of 2SNG-QMLE is to linearly combine candidate quasi likelihoods in the second step. Instead of choosing a single likelihood, the log-likelihood objective in the second step is a weighted average of candidate log-likelihoods. The weights are chosen adaptively to optimize the asymptotic variance. By such combination efficiency, it will cover more dimensions of innovation distributions, and the efficiency will be further improved.

\appendix

\section{Appendix section}\label{app}
\subsection{Proof of Lemma \ref{lem:identification}}
\begin{proof}
\begin{eqnarray*}
E_t(l_t(\boldsymbol{\theta}))&=&Q(\frac{\eta_f\sigma
v_t(\boldsymbol{\gamma})}{\sigma_0 v_t(\boldsymbol{\gamma_0})})-\log \sigma_0
v_t(\boldsymbol{\gamma_0})+\log\eta_f\\&\leq&
Q(\eta_f)-\log \sigma_0
v_t(\boldsymbol{\gamma_0})+\log\eta_f\\&=&E_t(l_t(\boldsymbol{\theta_0}))
\end{eqnarray*}
By Assumption \ref{ASS1}, the inequality holds with positive
probability. Therefore, by iterated expectations,
$\bar{L}_T(\boldsymbol{\theta})<\bar{L}_T(\boldsymbol{\theta_0})$.
\end{proof}

\subsection{Proof of Lemma \ref{lem:primitive}}
\begin{proof}
Given regularity conditions, we have
$\dot{Q}(\eta)=-\frac{1}{\eta}E(1+h(\frac{\varepsilon}{\eta}))$. Denote $H(\eta)=E(1+h(\frac{\varepsilon}{\eta}))$.
$\ddot{Q}(\eta)=\frac{1}{\eta^2}H(\eta)-\frac{1}{\eta}\dot{H}(\eta)$, where
$\dot{H}(\eta)=-\frac{1}{\eta^2}E(\varepsilon
\dot{h}(\frac{\varepsilon}{\eta}))$, for any $\eta>0$. Because
$E(\varepsilon \dot{h}(\frac{\varepsilon}{\eta}))<0$, so $\dot{H}(\eta)>0$.
Next, $\lim_{\eta\rightarrow +\infty}H(\eta)=1$, since
\begin{eqnarray*}
\lim_{\eta\rightarrow +\infty}|H(\eta)-1|=\lim_{\eta\rightarrow
+\infty}|E(h(\frac{\varepsilon}{\eta}))|\leq \lim_{\eta\rightarrow
+\infty}\frac{E|\varepsilon|^p}{\eta^p}\rightarrow 0
\end{eqnarray*}
On the other hand, by Fatou's lemma along with 1 and 4, we have
\begin{eqnarray*}
\limsup_{\eta\rightarrow 0+}H(\eta)=\limsup_{\eta\rightarrow
0+}E(1+h(\frac{\varepsilon}{\eta}))\leq1+E(\limsup_{\eta\rightarrow
0+}h(\frac{\varepsilon}{\eta}))<0
\end{eqnarray*}
then $\lim_{\eta\rightarrow +\infty}H(\eta)=1$,
$\limsup_{\eta\rightarrow 0+}H(\eta)<0$, and $\dot{H}(\eta)>0$. Hence,
there exists a unique constant $\eta_f\in(0, \infty)$ such that
$H(\eta_f)=0$, hence $\dot{Q}(\eta_f)=0$, and $\ddot{Q}(\eta_f)<0$. This
concludes the proof.
\end{proof}

\subsection{Proof of Theorem \ref{thm:Con1}}
\begin{proof} The proof is similar to \cite{Elie_Jeantheau_1995} by verifying the conditions given in \cite{Pfanzagl1969}.
\end{proof}

\subsection{Proof of Theorem \ref{thm:CLT1}}
\begin{proof}
Let $\rho_t(\boldsymbol{\theta})=(\sigma v_t)^{-2}\boldsymbol{k}$ and $\sigma_t(\boldsymbol{\theta})=\sigma v_t(\boldsymbol{\gamma})$. Define the vector-valued function $\boldsymbol{\psi}$ as
\begin{eqnarray}\label{samplescore}
\boldsymbol{\psi}(\omega,\boldsymbol{\theta})&=&\frac{\partial L_T}{\partial \boldsymbol{\theta}}=-\frac{1}{T}\sum_{t=1}^T(1+\frac{\dot{f}(\frac{x_t}{\eta_f\sigma_t})}{f(\frac{x_t}{\eta_f\sigma_t})}\frac{x_t}{\eta_f\sigma_t})\boldsymbol{k}\nonumber
\end{eqnarray}
For convenience, we consider the parameters ranging within a local neighborhood of the true values as in \cite{HY2003}. This simplification may not be critical, given that the estimator is proved to be consistent. By Taylor expansion,

\begin{eqnarray}
\sigma_t(\boldsymbol{\theta})^2&=&\sigma_t(\boldsymbol{\theta_0})^2+\boldsymbol{A_t}(\boldsymbol{\theta_0})'(\boldsymbol{\theta}-\boldsymbol{\theta_0})+\parallel\boldsymbol{\theta}-\boldsymbol{\theta_0}
\parallel^2\boldsymbol{R_{1t}}(\boldsymbol{\theta})\sigma_t(\boldsymbol{\theta_0})^2\\
\boldsymbol{\rho_t}(\boldsymbol{\theta})&=&\boldsymbol{\rho_t}(\boldsymbol{\theta_0})+\boldsymbol{B_t}(\boldsymbol{\theta_0})(\boldsymbol{\theta}-\boldsymbol{\theta_0})+\parallel\boldsymbol{\theta}-
\boldsymbol{\theta_0}\parallel^2\boldsymbol{R_{2t}}(\boldsymbol{\theta})\sigma_t(\boldsymbol{\theta_0})^{-2}
\end{eqnarray}
where $\boldsymbol{R_{1t}}(\boldsymbol{\theta})$ and $\boldsymbol{R_{2t}}(\boldsymbol{\theta})$ are an $r$-vector and
$r\times r$ matrix, and $r=1+p+q$.

On the other hand,
\begin{eqnarray}
&&h(\frac{x_t}{\eta_f\sigma_t})=h(\frac{\varepsilon_t\sigma_t(\boldsymbol{\theta_0})}{\eta_f\sigma_t(\boldsymbol{\theta})})\nonumber\\
&=&h(\frac{\varepsilon_t}{\eta_f})-\frac{\varepsilon_t}{\eta_f}\dot{h}(\frac{\varepsilon_t}{\eta_f})\sigma_t(\boldsymbol{\theta_0})^2\boldsymbol{\rho_t}
(\boldsymbol{\theta_0})'(\boldsymbol{\theta}-\boldsymbol{\theta_0})+
\parallel\boldsymbol{\theta}-\boldsymbol{\theta_0}\parallel^2\boldsymbol{R_{3t}}(\boldsymbol{\theta})
\end{eqnarray}
where $\boldsymbol{R_{3t}}(\boldsymbol{\theta})$ is an $r$-vector.

It has been shown in \cite{HY2003} that for $\boldsymbol{R_t}(\boldsymbol{\theta})=\boldsymbol{R_{1t}}(\boldsymbol{\theta})$,
$\boldsymbol{R_{2t}}(\boldsymbol{\theta})$ and $\boldsymbol{R_{3t}}(\boldsymbol{\theta})$, component-wise,
\begin{eqnarray}
P(T^{-1}\sum_{t=1}^T\sup_{|\boldsymbol{\theta}-\boldsymbol{\theta_0}|\leq \xi }|\boldsymbol{R_t}(\boldsymbol{\theta})|\leq C)\longrightarrow 1
\end{eqnarray}
with $\xi$ sufficiently small.
Therefore, we can rewrite the equation (\ref{samplescore}) as
\begin{eqnarray}
&&\boldsymbol{0}=\sum_{t=1}^T(1+h(\frac{\varepsilon_t}{\eta_f}))\sigma_t(\boldsymbol{\theta_0})^2\boldsymbol{\rho_t}(\boldsymbol{\theta_0})+
\sum_{t=1}^T\big(-\frac{\varepsilon_t}{\eta_f}\dot{h}(\frac{\varepsilon_t}{\eta_f})\sigma_t(\boldsymbol{\theta_0})^4\boldsymbol{\rho_t}(\boldsymbol{\theta_0})'
\boldsymbol{\rho_t}(\boldsymbol{\theta_0})\nonumber\\&&+(1+h(\frac{\varepsilon_t}{\eta_f}))
(\boldsymbol{A_t}(\boldsymbol{\theta_0})'\boldsymbol{\rho_t}(\boldsymbol{\theta_0})+\sigma_t(\boldsymbol{\theta_0})^2\boldsymbol{B_t}(\boldsymbol{\theta_0}))
\big)(\boldsymbol{\theta}-\boldsymbol{\theta_0})+\parallel\boldsymbol{\theta}-\boldsymbol{\theta_0}\parallel^2T\boldsymbol{R}(\boldsymbol{\theta})\nonumber
\end{eqnarray}
where
\begin{eqnarray}
P(\sup_{|\boldsymbol{\theta}-\boldsymbol{\theta_0}|\leq \xi }|\boldsymbol{R}(\boldsymbol{\theta})|\leq C)\longrightarrow 1
\end{eqnarray}
Note that
\begin{eqnarray}
&&E\Big((1+h(\frac{\varepsilon_t}{\eta_f}))(\boldsymbol{A_t}(\boldsymbol{\theta_0})'\boldsymbol{\rho_t}(\boldsymbol{\theta_0})+\sigma_t(\boldsymbol{\theta_0})^2
\boldsymbol{B_t}(\boldsymbol{\theta_0}))\Big)\nonumber\\
&=&E\Big((\boldsymbol{A_t}(\boldsymbol{\theta_0})'\boldsymbol{\rho_t}(\boldsymbol{\theta_0})+\sigma_t(\boldsymbol{\theta_0})^2\boldsymbol{B_t}(\boldsymbol{\theta_0}))
E_t\Big(1+h(\frac{\varepsilon_t}{\eta_f})\Big)\Big)\nonumber\\
&=&E\Big((\boldsymbol{A_t}(\boldsymbol{\theta_0})'\boldsymbol{\rho_t}(\boldsymbol{\theta_0})+\sigma_t(\boldsymbol{\theta_0})^2\boldsymbol{B_t}(\boldsymbol{\theta_0}
))\Big)E\Big(1+h(\frac{\varepsilon_t}{\eta_f})\Big)\nonumber\\
&=&\boldsymbol{0}
\end{eqnarray}
Therefore, it may be proved from the ergodic theorem that
\begin{eqnarray}
&&T^{-1}\sum_{t=1}^T(1+h(\frac{\varepsilon_t}{\eta_f})(\boldsymbol{A_t}(\boldsymbol{\theta_0})'\boldsymbol{\rho_t}(\boldsymbol{\theta_0})+
\sigma_t(\boldsymbol{\theta_0})^2\boldsymbol{B_t}(\boldsymbol{\theta_0}))\longrightarrow \boldsymbol{0}\\
&&T^{-1}\sum_{t=1}^T\frac{\varepsilon_t}{\eta_f}\dot{h}(\frac{\varepsilon_t}{\eta_f})\sigma_t(\boldsymbol{\theta_0})^4\boldsymbol{\rho_t}(\boldsymbol{\theta_0}
)'\boldsymbol{\rho_t}(\boldsymbol{\theta_0})\longrightarrow
\boldsymbol{M}
E\Big(\frac{\varepsilon_t}{\eta_f}\dot{h}(\frac{\varepsilon_t}{\eta_f})\Big)
\end{eqnarray}
Hence, we have
\begin{eqnarray}
&&(\boldsymbol{M}
E\Big(\frac{\varepsilon_t}{\eta_f}\dot{h}(\frac{\varepsilon_t}{\eta_f})\Big)+o_P(1))(\boldsymbol{\theta}-\boldsymbol{\theta_0})
+\parallel\boldsymbol{\theta}-\boldsymbol{\theta_0}\parallel^2\boldsymbol{R}(\boldsymbol{\theta})\nonumber\\
&=&T^{-1}\sum_{t=1}^T(1+h(\frac{\varepsilon_t}{\eta_f}))\boldsymbol{k_0}
\end{eqnarray}
where $o_p(1)$ does not depend on $\boldsymbol{\theta}$. It may be proved from martingale central limit theorem that
\begin{eqnarray}
T^{-1/2}\sum_{t=1}^T(1+h(\frac{\varepsilon_t}{\eta_f}))\boldsymbol{k_0})\stackrel{\rm
\mathcal{L}}{\longrightarrow}N\Big(0,\boldsymbol{M}
E(1+h(\frac{\varepsilon_t}{\eta_f}))^2\Big)
\end{eqnarray}
then it follows from the same argument as in \cite{HY2003} that
\begin{eqnarray}
\hat{\boldsymbol{\theta}}_{\boldsymbol{T}}-\boldsymbol{\theta_0}=\boldsymbol{O_p(T^{-1/2})}
\end{eqnarray}
and
\begin{eqnarray}
(\boldsymbol{M}
E\Big(\frac{\varepsilon_t}{\eta_f}\dot{h}(\frac{\varepsilon_t}{\eta_f})\Big)+o_P(1))(\boldsymbol{\theta}-\boldsymbol{\theta_0})
=T^{-1}\sum_{t=1}^T(1+h(\frac{\varepsilon_t}{\eta_f}))\boldsymbol{k_0}
\end{eqnarray}
Thus,
\begin{eqnarray}
T^{1/2}(\hat{\boldsymbol{\theta}}_{\boldsymbol{T}}-\boldsymbol{\theta_0})\stackrel{\rm
\mathcal{L}}{\longrightarrow}N\Big(\boldsymbol{0},\boldsymbol{M}^{-1}
\frac{E\Big((1+h(\frac{\varepsilon_t}{\eta_f}))^2\Big)}{\Big(E\Big(\frac{\varepsilon_t}{\eta_f}\dot{h}(\frac{\varepsilon_t}{\eta_f})\Big)\Big)^2}\Big)
\end{eqnarray}
\end{proof}

\subsection{Proof of Proposition \ref{cor:example}}
\begin{proof}
Define the likelihood ratio function
$G(\eta)=E(\log(\frac{\frac{1}{\eta}f(\frac{\varepsilon}{\eta})}{f(\varepsilon)}))$.
Suppose $G(\eta)$ has no local extremal values. And since
$\log(x)\leq2(\sqrt{x}-1)$,
\begin{eqnarray*}
E(\log(\frac{\frac{1}{\eta}f(\frac{\varepsilon}{\eta})}{f(\varepsilon)}))&\leq&2E(\sqrt{\frac{\frac{1}{\eta}f(\frac{\varepsilon}{\eta})}{f(\varepsilon)}}-1)
=2\int_{-\infty}^{+\infty}\sqrt{\frac{1}{\eta}f(\frac{x}{\eta})f(x)}dx-2\\
&\leq&-\int_{-\infty}^{+\infty}(\sqrt{\frac{1}{\eta}f(\frac{x}{\eta})}-\sqrt{f(x)})^2dx\\
&\leq&0
\end{eqnarray*}
The equality holds if and only if $\eta=1$. Therefore, $\eta=1$ is
the unique maximum of $Q(\eta)$.

\end{proof}

\subsection{Proof of Theorem \ref{thm:CLT2}}
In order to show the asymptotic normality, we first list some notations and derive a lemma. For convenience, we denote $\boldsymbol{y_0} = {1 \over v_t(\boldsymbol{\gamma_0})} {\partial v_t(\boldsymbol{\gamma_0})\over \partial \boldsymbol{\gamma}} $ and $\bar {\boldsymbol{y}}_{\boldsymbol{0}} = E(\boldsymbol{y_0})$, so $\boldsymbol{k_0} = ({1\over \sigma_0}, \boldsymbol{y_0}')'$ and $\bar{ \boldsymbol{k}}_{\boldsymbol{0}} = E {\boldsymbol{k_0}} = ({1\over \sigma_0}, \bar{\boldsymbol{ y}}_{\boldsymbol{0}}')'$. Also, let $\boldsymbol{M} = E(\boldsymbol{k_0}\boldsymbol{k_0}')$, $\boldsymbol{N} = \bar{\boldsymbol{ k}}_{\boldsymbol{0}} \bar {\boldsymbol{k}}_{\boldsymbol{0}}'$ and $\boldsymbol{V}= \mbox{Var}(\boldsymbol{y_0})^{-1}$. All the expectations above are taken under the true density $g$.

\begin{Lemma}\label{lem:M^-1}The following claims hold:
\begin{enumerate}
\item The inverse of $\boldsymbol{M}$ in block expression is
\begin{eqnarray}\label{eq:M-1 in lemma}
\boldsymbol{M}^{-1} = \left( \begin{array}{cc}
\sigma_0^2 (1+ \bar{ \boldsymbol{y}}_{\boldsymbol{0}}' \boldsymbol{V} \bar {\boldsymbol{y}}_{\boldsymbol{0}}) & -\sigma_0 \bar {\boldsymbol{y}}_{\boldsymbol{0}}' {\boldsymbol{V}} \\
-\sigma_0 {\boldsymbol{V}}\bar{\boldsymbol{ y}}_{\boldsymbol{0}} &{\boldsymbol{V}}
\end{array}\right );
\end{eqnarray}
\item $\bar{\boldsymbol{k}}_{\boldsymbol{0}}'{\boldsymbol{M}}^{-1}=\sigma_0\boldsymbol{e_1}'$, $\bar{\boldsymbol{ k}}_{\boldsymbol{0}}' {\boldsymbol{M}}^{-1} {\boldsymbol{ k}}_{\boldsymbol{0}} = \bar {\boldsymbol{k}}_{\boldsymbol{0}}' {\boldsymbol{M}}^{-1} \bar {\boldsymbol{k}}_{\boldsymbol{0}} = 1$;
\item $\boldsymbol{M}^{-1}\boldsymbol{N}\boldsymbol{M}^{-1} =\boldsymbol{M}^{-1}\boldsymbol{N}\boldsymbol{M}^{-1}\boldsymbol{N}\boldsymbol{M}^{-1} = \sigma_0^2 \boldsymbol{e_1} \boldsymbol{e_1}'$, where $\boldsymbol{e_1}$ is a unit column vector that has the same length as $\boldsymbol{\theta}$, with the first entry one and all the rest zeros.
\end{enumerate}
\end{Lemma}
\begin{proof}
The proof uses Moore-Penrose pseudo inverse described in \cite{Adi2003}. Observe that
\begin{eqnarray}
\boldsymbol{M} = \left( \begin{array}{cc}
0 & \boldsymbol{0} \\
\boldsymbol{0} & \mbox{Var}(\boldsymbol{y_0})
\end{array} \right) + \bar {\boldsymbol{k}}_{\boldsymbol{0}} \bar {\boldsymbol{k}}_{\boldsymbol{0}}'.
\end{eqnarray}
Use the technique of Moore-Penrose pseudo inverse,
\begin{eqnarray}\label{eq:M-1 in proof_lemma}
\boldsymbol{M}^{-1} = \left( \begin{array}{cc}
0 & \boldsymbol{0} \\
\boldsymbol{0} & \mbox{Var}(\boldsymbol{y_0})
\end{array} \right)^+ + \boldsymbol{H} = \left( \begin{array}{cc}
0 & \boldsymbol{0}\\
\boldsymbol{0} & \boldsymbol{V}
\end{array}\right) + \boldsymbol{H}
\end{eqnarray}
where $H$ is formed by the elements below:
\begin{eqnarray*}
\beta &=& 1 + \bar {\boldsymbol{y}}_{\boldsymbol{0}}' \boldsymbol{V} \bar {\boldsymbol{y}}_{\boldsymbol{0}} \\
\boldsymbol{w} &=& ({\sigma^{-1}_0}, \boldsymbol{0})' \\
\boldsymbol{m} &=&\boldsymbol{ w} \\
\boldsymbol{v} &=& (0, \bar {\boldsymbol{y}}_{\boldsymbol{0}}' \boldsymbol{V})' \\
\boldsymbol{n} &=&\boldsymbol{ v}
\end{eqnarray*}
\begin{eqnarray*}
\boldsymbol{ H}&=& - {1\over \| \boldsymbol{w} \|^2} \boldsymbol{v} \boldsymbol{w}' - {1\over \|\boldsymbol{m}\|^2} \boldsymbol{m} \boldsymbol{n}' + {\beta \over \|\boldsymbol{w}\|^2 \|\boldsymbol{m}\|^2} \boldsymbol{m}\boldsymbol{ w}'\\
&=& \sigma_0^2 \left( \begin{array}{cc}
1+\bar{\boldsymbol{ y}}_{\boldsymbol{0}}' \boldsymbol{V} \bar {\boldsymbol{y}}_{\boldsymbol{0}} & -{\sigma^{-1}_0} \bar{\boldsymbol{ y}}_{\boldsymbol{0}}' {\boldsymbol{V}} \\
-{\sigma^{-1}_0}{\boldsymbol{V}}\bar {\boldsymbol{y}}_{\boldsymbol{0}} & \boldsymbol{0}
\end{array}\right)
\end{eqnarray*}
So (\ref{eq:M-1 in lemma}) is obtained by plugging $\boldsymbol{H}$ into (\ref{eq:M-1 in proof_lemma}). The rest two points of the lemma can be obtained by simple matrix manipulation.
\end{proof}

Next we return to the proof of Theorem \ref{thm:CLT2}.

\begin{proof}
According to Theorem 3.4 in \cite{NeweyMcFadden}, $(\tilde {\boldsymbol{\theta}}_{\boldsymbol{T}}, \hat \eta, \hat{\boldsymbol{ \theta}}_{\boldsymbol{T}})$ are jointly $T^{1\over 2}$-consistent and asymptotic normal. The asymptotic variance matrix is
\begin{eqnarray}
\boldsymbol{G}^{-1}E(\tilde{ \boldsymbol{s}}(\boldsymbol{\theta_0}, \eta_f, \boldsymbol{\theta_0}) \tilde {\boldsymbol{s}}(\boldsymbol{\theta_0}, \eta_f, \boldsymbol{\theta_0})')\boldsymbol{G}'^{-1}.
\end{eqnarray}
where $\boldsymbol{G} = E(\nabla \tilde{\boldsymbol{ s}}(x_t,\boldsymbol{\theta_0}, \eta_f, \boldsymbol{\theta_0}))$.
View this matrix as $3\times 3$ blocks, with asymptotic variances of $(\tilde {\boldsymbol{\theta}}_{\boldsymbol{T}}, \hat \eta, \boldsymbol{\hat \theta}_{\boldsymbol{T}})$ on the first, second and third diagonal blocks. We now calculate the second and the third diagonal blocks. The expect Jacobian matrix $G $ can be decomposed into
\begin{eqnarray*}
\boldsymbol{G} =
E \left(
\begin{array} {ccc}
\nabla_{\boldsymbol{\theta}}\boldsymbol{s_1}(x_t,\boldsymbol{\theta_0}) & \boldsymbol{0} & \boldsymbol{0} \\
\nabla_{\boldsymbol{\theta}}s_2(x_t,\boldsymbol{\theta_0}, \eta_f) & \nabla_{\eta} s_2 (\boldsymbol{\theta_0}, \eta_f) &\boldsymbol{0}\\
\boldsymbol{0} & \nabla_{\eta} \boldsymbol{s_3}(x_t, \eta_f, \boldsymbol{\theta_0}) & \nabla_{\boldsymbol{\phi}} \boldsymbol{s_3}(x_t, \eta_f, \boldsymbol{\theta_0})
\end{array}  \right)
\end{eqnarray*}
Denote the corresponding blocks as $G_{ij}$, $i,j=1,2,3.$ Direct calculation yields
\begin{eqnarray*}
\boldsymbol{G_{11}} &=& -2 \boldsymbol{M} \\
\boldsymbol{G_{21}} &=& {1\over \eta_f} E \Big( \dot{h}({\varepsilon \over \eta_f}) {\varepsilon \over \eta_f} \Big) \bar{\boldsymbol{ k}}_{\boldsymbol{0}}' \\
G_{22} &=& {1\over \eta_f^2} E \Big( \dot{h}({\varepsilon \over \eta_f}) {\varepsilon \over \eta_f} \Big)\\
\boldsymbol{G_{32}} &=& \boldsymbol{G_{21}}' \\
\boldsymbol{G_{33}} &=& E \Big( \dot{h}({\varepsilon \over \eta_f}) {\varepsilon \over \eta_f} \Big) \boldsymbol{M}
\end{eqnarray*}
The second diagonal block depends on the second row of $\boldsymbol{G^{-1}}$ and $\tilde {\boldsymbol{s}}(x_t, \boldsymbol{\theta_0}, \eta_f, \boldsymbol{\theta_0})$. The second row of $\boldsymbol{G}^{-1}$ is
\begin{eqnarray*}
(-G_{22}^{-1}\boldsymbol{G_{21}}\boldsymbol{G_{11}}^{-1} \quad G_{22}^{-1} \quad \boldsymbol{0})
\end{eqnarray*}
So the asymptotic variance of $\hat \eta$ is
$G_{22}^{-1} E(\boldsymbol{q_2} \boldsymbol{q}'_{\boldsymbol{2}}) G_{22}'^{-1}$, where
\begin{eqnarray*}
{q_2}     &=& -\boldsymbol{G_{21}}\boldsymbol{G_{11}}^{-1} \boldsymbol{s_1}(x_t, \boldsymbol{\theta_0}) + s_2(x_t, \boldsymbol{\theta_0}, \eta_f)\\
&=& {1\over 2 \eta_f} E \Big( \dot{h}({\varepsilon \over \eta_f}) {\varepsilon \over \eta_f} \Big) \bar{\boldsymbol{ k}}_{\boldsymbol{0}}' \overline{\boldsymbol{k_0}\boldsymbol{k_0}'}^{-1} (\varepsilon^2 -1)\boldsymbol{k_0} - {1\over \eta_f} \Big( 1+ h({\varepsilon \over \eta_f}) \Big)\\
&=& {1\over \eta_f} \Big({ 1\over 2} {Eh_2(\varepsilon^2 -1)} - h_1 \Big)
\end{eqnarray*}
The last step uses the second point of Lemma \ref{lem:M^-1}. So (\ref{eq:avar_eta}) is obtained by plugging in the expressions for $G_{22}$ and ${q_2}$. Similarly, the third row of $\boldsymbol{G^{-1}}$ is
\begin{eqnarray*}
\boldsymbol{G_{33}}^{-1}(\boldsymbol{G_{32}}G_{22}^{-1}\boldsymbol{G_{21}}\boldsymbol{G_{11}}^{-1} \quad -\boldsymbol{G_{32}}G_{22}^{-1} \quad \boldsymbol{I})
\end{eqnarray*}
The asymptotic variance for $\hat {\boldsymbol{\theta}}$ is
$ \boldsymbol{G_{33}}^{-1} E(\boldsymbol{q_3} \boldsymbol{q}'_{\boldsymbol{3}}) \boldsymbol{G_{33}}'^{-1}$,
where
\begin{eqnarray*}
\boldsymbol{q_3} &=&\boldsymbol{ G_{32}}G_{22}^{-1} (\boldsymbol{G_{21}}\boldsymbol{G_{11}}^{-1}\boldsymbol{s_1}(x_t, \boldsymbol{\theta_0}) - \boldsymbol{s_2}(x_t, \boldsymbol{\theta_0}, \eta_f) ) + \boldsymbol{s_3}(x_t, \eta_f, \boldsymbol{\theta_0}) \\
&=& -(1+h({\varepsilon \over \eta_f}))(\boldsymbol{k_0}-\bar{\boldsymbol{ k}}_{\boldsymbol{0}}) - {1\over2} E(\dot{h}({\varepsilon \over \eta_f}){\varepsilon \over \eta_f})\bar{\boldsymbol{ k}}_{\boldsymbol{0}} \bar {\boldsymbol{k}}_{\boldsymbol{0}}' (\overline{{\boldsymbol{k_0}}\boldsymbol{k_0}'})^{-1}\boldsymbol{k_0} (\varepsilon^2-1)\\
&=& -h_1(\boldsymbol{k_0}-\bar {\boldsymbol{k}}_{\boldsymbol{0}}) - {1\over 2} (Eh_2) (\varepsilon^2-1)\bar{\boldsymbol{ k}}_{\boldsymbol{0}}
\end{eqnarray*}
The last step uses the second point of Lemma \ref{lem:M^-1}. Then
\begin{eqnarray*}
E\boldsymbol{q_3}\boldsymbol{q_3}' &=& Eh_1^2 (\boldsymbol{M-N}) + {1\over 4}(Eh_2)^2 E(\varepsilon^2-1)\boldsymbol{N} \\
&=& Eh_1^2 \boldsymbol{M} + \Big({1\over 4}E(\varepsilon^2-1)^2 - Eh_1^2 \Big)\boldsymbol{N}
\end{eqnarray*}
Therefore, (\ref{eq:avar_hat theta}) is obtained by plugging in the expressions for $\boldsymbol{G_{33}}$, $E\boldsymbol{q_3}\boldsymbol{q_3}'$, and apply the third point of Lemma \ref{lem:M^-1}.

The asymptotic covariance between $\hat {\boldsymbol{\theta}}$ and $\hat \eta_f$ is $\boldsymbol{G_{33}}^{-1}E(\boldsymbol{q_3}{ q}_{{2}})G_{22}'^{-1}$, then direct calculation using the second point of Lemma \ref{lem:M^-1} yields
\begin{eqnarray*}
\boldsymbol{\Pi}={\eta_f\sigma_0 \over 2}E\Big((\varepsilon^2-1)({h_1\over Eh_2}-{\varepsilon^2-1\over 2})\Big)\boldsymbol{e'_1}
\end{eqnarray*}
The same formula recurs in the asymptotic covariance between $\tilde {\boldsymbol{\theta}}$ and $\hat \eta_f$, which is $\boldsymbol{G_{11}}^{-1}E(\boldsymbol{q_1}{ q}_{{2}})G_{22}'^{-1}$.

Finally, the asymptotic covariance between $\tilde {\boldsymbol{\theta}}$ and $\hat {\boldsymbol{\theta}}$ is $\boldsymbol{G_{11}}^{-1}E(\boldsymbol{q_1}{\boldsymbol q}'_{\boldsymbol{3}})\boldsymbol{G_{33}}'^{-1}$, denoted as $\Xi$. If implies from the third point of Lemma \ref{lem:M^-1} that
\begin{eqnarray*}
\boldsymbol{\Xi}=\frac{E(h_1(\varepsilon^2-1))}{2E(h_2)}\boldsymbol{M^{-1}}-\frac{\sigma^2_0}{2}E\Big((\varepsilon^2-1)({h_1\over Eh_2}-{\varepsilon^2-1\over 2})\Big)\boldsymbol{e_1}\boldsymbol{e'_1}
\end{eqnarray*}

which concludes the proof.

\end{proof}

\subsection{Proof of Theorem \ref{thm:CLT3}}
\begin{proof}
Following the similar idea to GMM, we may prove:
\begin{eqnarray*}
\left(
           \begin{array}{ccc}
             \boldsymbol{I} & \boldsymbol{0} &\boldsymbol{ 0} \\
             \lambda_T T^{-\frac{1}{2}}\boldsymbol{G_{21}} & G_{22} & \boldsymbol{0} \\
            \boldsymbol{0}  & \boldsymbol{G_{32}} & \boldsymbol{G_{33}} \\
           \end{array}
         \right)
 \left(
\begin{array}{c}
T\lambda^{-1}_T(\tilde{\boldsymbol{\theta}}-\boldsymbol{\theta_0}) \\
T^{\frac{1}{2}}(\hat{\eta}_f-\eta_f) \\
  T^{\frac{1}{2}}(\hat{\boldsymbol{\theta}}-\boldsymbol{\theta_0})
\end{array}\right)
=\left(\begin{array}{c}
\frac{1}{\lambda_T}\sum_{t=1}^T\boldsymbol{\Psi_t}(\varepsilon_t)+\boldsymbol{o_P(1)} \\
   -\frac{1}{\sqrt{T}}\sum_{t=1}^T\frac{1}{\eta_f}(1+h(\frac{\varepsilon_t}{\eta_f}))+o_P(1) \\
   \frac{1}{\sqrt{T}}\sum_{t=1}^T (1+h(\frac{\varepsilon_t}{\eta_f}))\boldsymbol{k_0}+\boldsymbol{o_P(1)}
 \end{array}\right)
\end{eqnarray*}

Clearly, the corresponding weighting vector for $\sqrt{T}(\hat{\boldsymbol{\theta}}_{\boldsymbol{T}}-\boldsymbol{\theta_0})$ is
\begin{eqnarray*}\left(
 \begin{array}{ccc}
\boldsymbol{G_{33}}^{-1}\boldsymbol{G_{32}}G_{22}^{-1}\lambda_T T^{-\frac{1}{2}}\boldsymbol{G_{21}} & -\boldsymbol{G_{33}}^{-1}\boldsymbol{G_{32}}
G_{22}^{-1} & \boldsymbol{G_{33}}^{-1}\\
\end{array}\right)
\end{eqnarray*}
Note that
\begin{eqnarray*}
&&\boldsymbol{G_{33}}^{-1}\boldsymbol{G_{32}}G_{22}^{-1}\lambda_T T^{-\frac{1}{2}}\boldsymbol{G_{21}}=\lambda_T T^{-\frac{1}{2}} \boldsymbol{M}^{-1}\bar{\boldsymbol{k}}_{\boldsymbol{0}}\bar{\boldsymbol{k}}'_{\boldsymbol{0}}=\lambda_T T^{-\frac{1}{2}}\sigma_0\boldsymbol{e_1}\bar{{\boldsymbol{k}}}'_{\boldsymbol{0}}\\
&&-\boldsymbol{G_{33}}^{-1}\boldsymbol{G_{32}}G_{22}^{-1}=-\sigma_0\eta_f( E \Big( \dot{h}({\varepsilon \over \eta_f}) {\varepsilon \over \eta_f} \Big))^{-1}\boldsymbol{e_1}
\end{eqnarray*}
thus the sub-matrices corresponding to $\boldsymbol{\gamma}$ parameter are $\boldsymbol{0}$s. Therefore, the first step has no effect on the central limit theorem of $\hat{\boldsymbol{\gamma}}_{\boldsymbol{T}}$. The result follows from Lemma \ref{lem:M^-1}. In terms of $\hat\sigma_T$, its convergence rate becomes $T\lambda_T^{-1}$.
\end{proof}

\subsection{Proof of Proposition \ref{Prop:aggregation}}
\begin{proof}
 Denote random variables $\kappa_G = (1-\ve^2)/2$, and $\kappa_2 = h_1(\ve/\eta_f)/E(h_2(\ve/\eta_f))$. We show the optimal weights for $\sigma$ and $\Bgamma$ are the same. From Lemma \ref{lem:M^-1}, Theorem \ref{thm:CLT2} and (\ref{eq:aggre_weights}), for $\sigma$, the numerator in $w^*_1$ is
 \begin{eqnarray*}
 (\BSigma_{\BG})_{1,1} - \BXi_{1,1} &=& \sigma_0^2 (1+\bar \By_0' \BV \bar \By_0)E \kappa_G^2 - \sigma_0^2 E\kappa_G^2 + \sigma_0^2 \bar \By_0' \BV \bar \By_0 E(\kappa_G\kappa_2) \\
 &=& \sigma_0^2 \bar \By_0' \BV \bar \By_0 E(\kappa_G (\kappa_G + \kappa_2))
 \end{eqnarray*}
 The denominator in $w^*_1$ is
 \begin{eqnarray*}
 (\BSigma_{\BG})_{1,1} +  (\BSigma_{\boldsymbol{2}})_{1,1}- 2 \BXi_{1,1} &=& \sigma_0^2 (1+\bar \By_0' \BV \bar \By_0)( E\kappa_G^2 + E\kappa_2^2) + \sigma_0^2 (E\kappa_G^2 - E\kappa_2^2) \\
   && - 2 \sigma_0^2 E\kappa_G^2 + 2\sigma_0^2 \bar \By_0' \BV \bar \By_0 E(\kappa_G \kappa_2)\\
   &=& \sigma_0^2 \bar \By_0' \BV \bar \By_0 E(\kappa_G^2  + \kappa_2^2 + 2\kappa_G \kappa_2)
   \end{eqnarray*}
 Therefore we obtain $w^*_1 = E(\kappa_G(\kappa_G + \kappa_2))/E(\kappa_G + \kappa_2)^2$. Now we compute the weights corresponding to $\Bgamma$. For $i = 2,\ldots,1+p+q$, let $j=i-1$, also from (\ref{eq:aggre_weights}),
 \begin{eqnarray*}
 w_i^* = {\BV_{j,j}E\kappa_G^2 + \BV_{j,j}E(\kappa_G \kappa_2)   \over \BV_{j,j}E\kappa_G^2 + \BV_{j,j}E\kappa_2^2  2\BV_{j,j}E(\kappa_G \kappa_2) } =
 { E(\kappa_G(\kappa_G + \kappa_2))  \over E(\kappa_G + \kappa_2)^2}
 \end{eqnarray*}
 Therefore all the optimal aggregation weights are the same.
\end{proof}

\bibliographystyle{ims}

\end{document}